%% file: paper.tex
\documentclass[letterpaper]{article} 
\usepackage{aaai25}  
\usepackage{times}  
\usepackage{helvet}  
\usepackage{courier}  
\usepackage[hyphens]{url}  
\usepackage{graphicx} 
\usepackage{natbib}  
\usepackage{caption} 
\usepackage{algorithm}
\usepackage[noend]{algpseudocode}
\usepackage{amsmath}
\usepackage{amsthm}
\usepackage{amssymb}
\usepackage{booktabs}
\usepackage{cleveref}
\usepackage{comment}
\usepackage{dsfont}
\usepackage[shortlabels]{enumitem}
\usepackage{fix-cm}
\usepackage{float}
\usepackage{graphicx}
\usepackage[utf8]{inputenc}
\usepackage[switch]{lineno}
\usepackage{multirow}
\usepackage{placeins}
\usepackage{soul}
\usepackage{subcaption}
\usepackage{tabularx}
\usepackage{tikz}
\usepackage{tikzsymbols}
\usetikzlibrary{backgrounds}
\usepackage{times}
\usepackage{todonotes}
\usepackage{url}
\usepackage{xcolor}
\usepackage{xfrac}

\title{Making Teams and Influencing Agents: Efficiently Coordinating Decision Trees for Interpretable Multi-Agent Reinforcement Learning}

\author{
Rex Chen,
Stephanie Milani,
Zhicheng Zhang,
Norman Sadeh,
Fei Fang
}
\affiliations{    
School of Computer Science, Carnegie Mellon University\\
rexc@cmu.edu,
smilani@andrew.cmu.edu,
zczhang@cmu.edu,
sadeh@cs.cmu.edu,
feifang@cmu.edu
}

\input{notation}
\newtheorem{theorem}{Theorem}

\begin{document}

\maketitle

\begin{abstract}
Poor interpretability hinders the practical applicability of multi-agent reinforcement learning (MARL) policies. Deploying interpretable surrogates of uninterpretable policies enhances the safety and verifiability of MARL for real-world applications. However, if these surrogates are to interact directly with the environment within human supervisory frameworks, they must be both performant and computationally efficient. Prior work on interpretable MARL has either sacrificed performance for computational efficiency or computational efficiency for performance. To address this issue, we propose HYDRAVIPER, a decision tree-based interpretable MARL algorithm. HYDRAVIPER coordinates training between agents based on expected team performance, and adaptively allocates budgets for environment interaction to improve computational efficiency. Experiments on standard benchmark environments for multi-agent coordination and traffic signal control show that HYDRAVIPER matches the performance of state-of-the-art methods using a fraction of the runtime, and that it maintains a Pareto frontier of performance for different interaction budgets.
\end{abstract}

\section{Introduction}
\label{sec:intro}
Over the past decade, \emph{multi-agent reinforcement learning} (MARL) algorithms have achieved state-of-the-art performance in various challenging board and video games \cite{Silver2016,Vinyals2019}. They have also found success in other sequential decision-making scenarios based on critical real-world domains, including robotics \cite{Orr2023}, cybersecurity \cite{Panfili2018}, and traffic signal control (TSC) \cite{Chen2020}. However, the real-world applicability of these algorithms is hampered by two key challenges. First, the deep neural network (\nn) architectures that are needed to achieve good performance have thousands to millions of parameters. Second, the behaviour of RL agents is difficult to predict and verify due to its dependence on complex state spaces and long time horizons. Thus, human stakeholders understand and trust RL agents less than their simpler counterparts, even if RL yields superior performance \cite{Siu2021}. Conversely, more interpretable representations of policies can help stakeholders build appropriate levels of trust in RL agents \cite{Druce2021,Zhang2020}. 

In applications where the safety and verifiability of RL policies is critical \cite{Gilbert2023,Jayawardana2021}, such as TSC, users may deploy interpretable surrogate policies instead of expert NN policies. Such surrogate policies should be \emph{performant} --- capable of achieving high returns. In MARL, coordinating the training of surrogates is critical for performance: if multiple surrogates are deployed simultaneously, they cannot assume that they are interacting with performant experts, as their performance may be influenced by other agents' suboptimal behaviour. 

At the same time, surrogate policies should be \emph{computationally efficient}; it should be possible to generate them with minimal environment interactions and runtime. Frequent interactions with the environment are impractical in domains such as robotics \cite{Finn2017} and TSC \cite{Zang2020}, and simulators that have sufficient fidelity for real-world transferrability \cite{Chen2023} are computationally intensive. In human-in-the-loop frameworks where users provide oversight to correct undesirable policy behaviour \cite{Mandel2017}, the ability to quickly iterate on surrogate policies is also critical \cite{Wu2023}. More complex models capable of stronger performance and coordination capabilities are less efficient \cite{Milani2024}. 

Decision trees ({\dt}s) are an attractive model class for interpretable RL due to their comprehensibility \cite{Silva2020}. They also enable the design of responsible AI systems, as their branching rules can be easily verified and constrained by human experts or automated processes \cite{Blockeel2023}. {\dt}s lie at the core of the imitation learning framework VIPER \cite{Bastani2018}, which has been applied to distil NN-based RL policies into {\dt}s in domains such as TSC \cite{Jayawardana2021}, autonomous vehicles \cite{Schmidt2021}, and robotics \cite{Roth2023}. However, generalising VIPER to the MARL setting is challenging. Past work \cite{Milani2022} introduced two multi-agent VIPER algorithms, IVIPER and MAVIPER, which are both impractical for deployment. IVIPER fails to coordinate the training of {\dt}s, thus sacrificing performance;  MAVIPER trains {\dt}s in a coordinated but computationally inefficient manner.

To this end, we introduce HYDRAVIPER\footnote{Code: \url{https://github.com/lythronaxargestes/hydraviper-public}}, an efficient method to extract coordinated {\dt} policies for cooperative MARL. Our method makes three key algorithmic contributions: (1) HYDRAVIPER coordinates agent training by jointly resampling the training dataset for each team of cooperative agents. (2) When interacting with the environment to collect a training dataset, HYDRAVIPER adaptively collects critical trajectories closer to convergence. (3) When interacting with the environment for evaluation, HYDRAVIPER uses a multi-armed bandit-based evaluation strategy to identify promising sets of trained surrogates. Experiments demonstrate that HYDRAVIPER achieves our goal of balancing performance and computational efficiency. HYDRAVIPER also improves the applicability of {\dt}-based interpretable MARL policies: users can exchange training time for performance by altering its environment interaction budgets, but its performance remains optimal at different budget levels. Lastly, HYDRAVIPER's efficiency on large environments can be improved while maintaining coordination by dividing the agent set into mutually influential teams. 

\section{Related Work}
\paragraph{Interpretable Multi-Agent Learning}
Past methods for interpretable MARL have focused on using feature importance measures to construct saliency maps \cite{Iqbal2019,Heuillet2022,Liu2023,Motokawa2023}, building logical structures \cite{Kazhdan2020,Wang2020a,Ji2023}, and defining domain concepts \cite{Zabounidis2023}. Each of these categories of methods has limitations. Feature importances and saliency maps are visually clear, but only highlight aspects of the state space without showing how policies use them. Policies based on logical rules and concepts allow users to align the execution of these policies with domain knowledge, but require extensive feature engineering. By contrast, we learn simple policy representations grounded directly in the environment feature space.

\paragraph{Decision Trees for Reinforcement Learning}
Relative to deep {\nn}s, shallow {\dt} policy representations are intrinsically \cite{Molnar2020} and empirically \cite{Silva2020} more comprehensible. One line of work in {\dt}-based RL directly trains {\dt} policies \cite{Silva2020,Topin2021,Crespi2023,Liu2025} using relaxations amenable to direct optimisation. However, these methods suffer from training instability and performance degradation. Another line of work follows the \emph{VIPER} framework \cite{Bastani2018}, in which a surrogate {\dt} is trained by imitation learning of a performant expert. Although VIPER has achieved success in single-agent settings \cite{Schmidt2021,Roth2023,Jayawardana2021,Zhu2022}, only two VIPER-based algorithms exist for multi-agent settings: \emph{IVIPER} and \emph{MAVIPER} \cite{Milani2022}. IVIPER independently trains {\dt}s for each agent in a decentralised manner, enjoying computational efficiency at the cost of performance due to a lack of coordination. MAVIPER jointly trains {\dt}s in a centralised manner to achieve coordination, but suffers from computational inefficiency. Neither algorithm balances performance and computational efficiency.

\section{Background}
\label{sec:background}
\begin{figure*}[ht]
    \centering
    \includegraphics[trim={0cm 3cm 0cm 2.5cm},clip,width=0.9\linewidth]{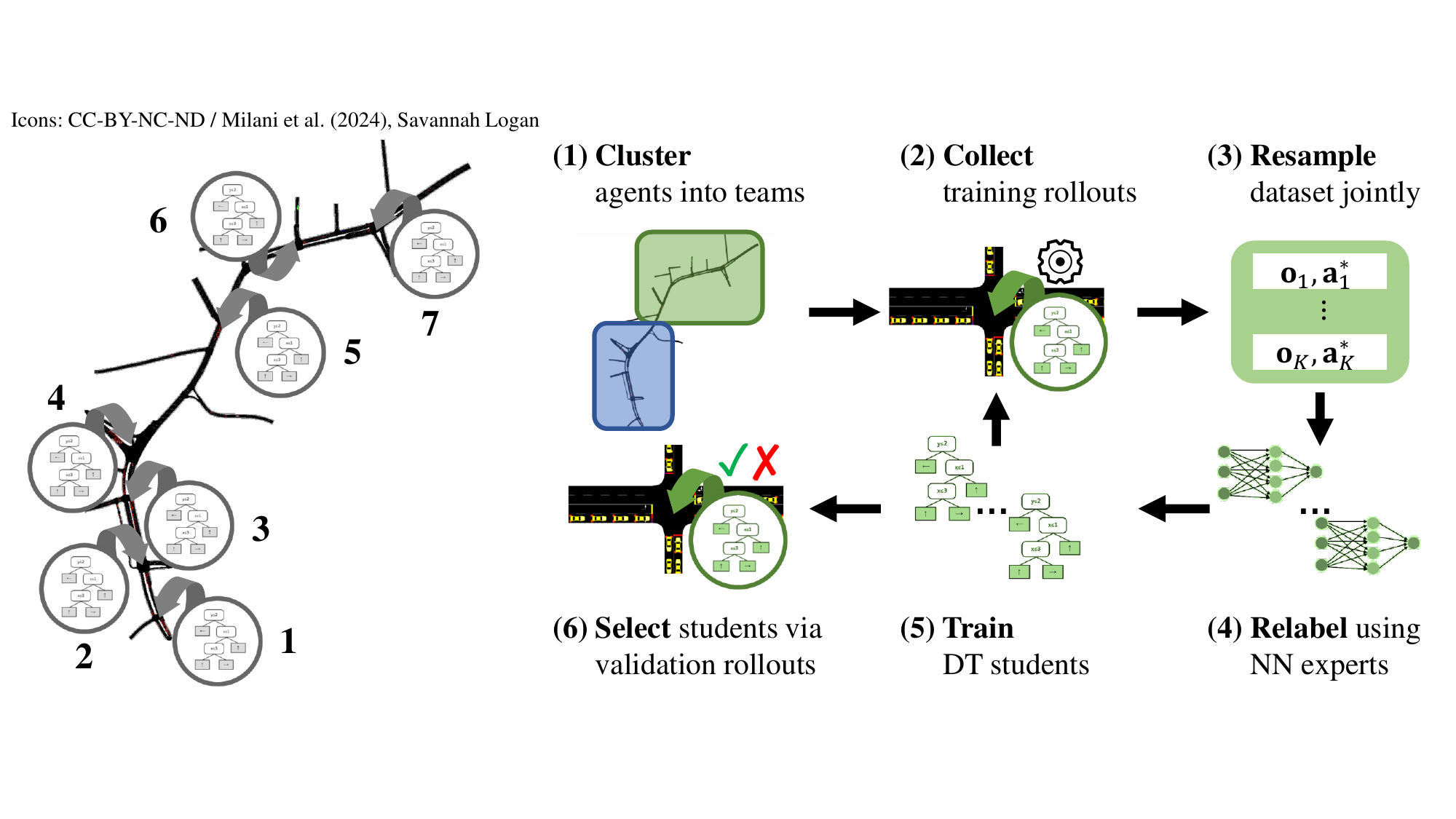}
    \caption{(L) Imitation learning in traffic signal control, where a decision tree must be learnt to imitate the RL-based policy of each intersection's signal controller agent (the seven-intersection Ingolstadt corridor TSC environment from \Cref{sec:exp:envs} is shown, with intersections numbered). (R) The HYDRAVIPER framework, in which {\dt} students are trained \emph{independently} using a \emph{jointly resampled} dataset of environment trajectories and relabelled by an NN expert.}
    \label{fig:algorithm}
\end{figure*}

\paragraph{Markov Games}
We model agent cooperation as a team-based \emph{Markov game}. A Markov game for $\nagents$ agents consists of a set of states $\statespace$ with initial state distribution $\rho: \statespace \rightarrow [0,1]$, and sets of actions $\allacts$ and observations $\allobs$ (consisting of features correlated with the state) for each agent $i$. The agent set is partitioned into disjoint teams $\team_1, \ldots, \team_L \subseteq [\nagents]$. Each agent chooses actions according to a policy $\pi_\aindex: \agentobsspace \rightarrow \agentactspace$. After agents simultaneously execute actions $\jointact$, the environment produces the next state based on the state transition function $\transitionprob: \statespace \times \actspace_1 \times \ldots \times \actspace_\nagents \rightarrow \statespace$, a joint reward for each team $\team_\ell$ based on $\rewardfunc_\ell : \statespace \times \prod_{\aindex \in \team_\ell} \actspace_\aindex \rightarrow \mathbb{R}$, and private observations $\obsfunc_\aindex : \statespace \rightarrow \agentobsspace$. Each agent maximises its team's return $\return_\ell = \sum_{t=0}^T \gamma^t \rewardfunc_\ell(s_t, \pi(\jointob_t))$ over $T$ timesteps, where $\gamma$ is a discount factor weighting the importance of future rewards.

\paragraph{MARL}
We refer to a policy profile as $\pi=(\pi_1, \ldots, \pi_\nagents)$, a policy profile excluding agent $\aindex$ as $\pi_{-\aindex}$, a joint policy profile for team $\team_\ell$ as $\pi_\ell = (\pi_i, \forall i \in \team_\ell)$, and a joint policy profile excluding all agents in $\team_\ell$ as $\pi_{-\ell} = (\pi_i, \forall i \notin \team_\ell)$. Each agent's \emph{value function} and \emph{state-action value} (or $Q$) \emph{function} characterise its expected returns under a policy profile $\pi$:
\begin{align*}
    &V^{\pi_\aindex} (s) = \agentreward + \gamma\sum_{s' \in \statespace} \transitionprob(\s, \pi_1(\ob_1), \ldots, \pi_{\nagents}(\ob_\nagents), \s')V^{\pi}_\aindex(\s'), \\
    &Q^{\pi_\aindex}(s, \jointact) = \agentreward + \gamma \sum_{s' \in S} \transitionprob(s, a_1, \ldots, a_\nagents,s')V^{\pi_\aindex}(s')
\end{align*}
Our algorithm assumes access to value and $Q$-functions that take the global observations of all agents, not states, as input: $V^{\pi_\aindex}(\jointob)$ and $Q^{\pi_\aindex}(\jointob, \jointact)$. Such $Q$-functions are used by actor-critic MARL algorithms that rely on centralised training with decentralised execution \cite{Lowe2017,Foerster2018}. Lastly, we define \emph{mean} value functions and state-action value functions for each team: $\bar{V}^{\pi_\ell}(\jointob) := \frac{1}{|\team_\ell|} \sum_{\aindex \in \team_\ell} V^{\pi_\aindex} (\jointob), \bar{Q}^{\pi_\ell}(\jointob, \jointact) := \frac{1}{|\team_\ell|} \sum_{\aindex \in \team_\ell} Q^{\pi_\aindex}(\jointob, \jointact)$.

\paragraph{Decision Trees}
A \emph{decision tree} (\dt) recursively partitions an input space $\mathcal{X}$ through functions $f_j: \mathcal{X} \rightarrow \mathbb{R}$ and thresholds $\tau_j$ at each internal node $j$. Each internal node induces a partition of samples, $P_j = {x \in \mathcal{X} : f_j(x) \leq \tau_j}$. For a {\dt} policy, internal nodes $(f_j, \tau_j)$ encode observation-dependent decision criteria, while leaf nodes $l \in \mathcal{L}$ map partitioned observations to actions: $\hat{\pi}_i: \mathcal{O}_i \rightarrow \mathcal{A}_i, \forall i \in [N]$.

\paragraph{VIPER}
VIPER \cite{Bastani2018} is an imitation learning framework, adapted from the more general DAGGER \cite{Ross2011}, that trains {\dt}s as surrogate policies. Given a trained \emph{expert} (\nn) policy $\pi^*$, VIPER iteratively generates \emph{student} (\dt) policies $\hat{\pi}^m$. Specifically, in each iteration $m$, VIPER:
\begin{enumerate}[(1)]
    \item \emph{Collects} $K$ new rollouts $\{\jointob, \hat{\pi}^{m-1}(\jointob)\}$ using the previous students from iteration $m-1$ (where $\hat{\pi}^0 := \pi^*$)
    \item \emph{Resamples} a dataset $\mathcal{D}$ from all trajectories collected so far, based on upweighting critical states where taking a suboptimal action may be costly in terms of $Q$-values:
    \begin{align*}
        p_k \propto V^{\pi^*}(\jointob_k) - \min_\jointact Q^{\pi^*}(\jointob_k, \jointact)
    \end{align*}
    \item \emph{Relabels} the dataset with the expert actions $\pi^*(\jointob_k)$
    \item \emph{Trains} new {\dt} students $\hat{\pi}^m$ on $\{\jointob_k, \pi^*(\jointob_k) \mid \jointob_k \in \mathcal{D}\}$
\end{enumerate}
After $M$ iterations, VIPER (5) \emph{selects} a student through validation on an additional set of rollouts. \citet{Ross2011} showed that such a procedure is guaranteed to find a student which is performant on the distribution of states that it induces.

\section{HYDRAVIPER}
\label{sec:hydraviper}
In this section, we present \emph{HYDRAVIPER} (\Cref{alg:hydraviper}), our algorithm for performant and efficient interpretable MARL. As shown in \Cref{fig:algorithm}, HYDRAVIPER builds on the DAGGER and VIPER frameworks by iteratively collecting data from environment rollouts to train {\dt} policies. HYDRAVIPER first (1) \emph{partitions agents} into clusters for scalability (line~4). Next, in each of $M$ iterations, HYDRAVIPER: (2) collects a dataset of rollouts from the environment, using an \emph{adaptive procedure} (lines~6--7); (3) resamples the dataset to prioritise learning the correct actions in critical states, using \emph{team-based $Q$-values} (lines~9--10); (4) and trains {\dt}s based on these datasets (lines~11--12). After it completes all $M$ training iterations, HYDRAVIPER (5) identifies the best-performing student for each agent, using a \emph{multi-armed bandit algorithm}, and returns them as a policy profile (lines~13--14). Now, we describe each of these algorithm components in detail.

\begin{algorithm}[ht]
\caption{HYDRAVIPER}
\label{alg:hydraviper}
\begin{algorithmic}[1]
\Statex \textbf{Input:} {\parbox[t]{0.85\linewidth}{%
    Markov game $(\statespace, \actspace, \transitionprob, \rewardfunc_\aindex, \obsfunc_\aindex)$, experts $\pi^*$, \\
    expert $Q$-functions $Q^{\pi^*}$, per-iteration rollout \\
    count $K_{\mathrm{train}}$, rollout budgets $(B_{\mathrm{train}}, B_{\mathrm{valid}})$, \\
    threshold $\epsilon$, iteration count $M$, scaling factor $c$, \\
    agent distance function $d$
    }
}
\Statex \textbf{Output:} Trained students $\hat{\pi}$

\State \textbf{Initialise} dataset $\mathcal{D} \gets \emptyset$, policies $\hat{\pi}_i^0 \gets \pi_i^*, \forall i \in N$
\State \textbf{Initialise} rollout count $n_{\mathrm{train}} \gets 0$
\State \textbf{Initialise} dropped rollout count $K_{\mathrm{drop}} \gets \infty$

\LineComment{\textcolor{red}{\Cref{sec:hydraviper:clustering}: Agent Clustering}}
\State \textbf{Cluster} agents $\team_1, \ldots, \team_L \gets \textbf{Partition}(\Gamma, \pi^*, d)$

\For{$m \in \{1, \dots, M\}$}
    \IndentLineComment{\textcolor{red}{\Cref{sec:hydraviper:training}: Training Rollouts}}
    \State $\mathcal{D}, n_{\mathrm{train}} \gets$ \textbf{TR-A}{\parbox[t]{0.55\linewidth}{%
        $(\mathcal{D}, \hat{\pi}^{m-1}, m, K_{\mathrm{train}}, B_{\mathrm{train}}$, \\
        $K_{\mathrm{drop}}, n_{\mathrm{train}})$
        }
    }

    \State \textbf{Reinitialise} dropped rollout count $K_{\mathrm{drop}} \gets \infty$
    \For{each team $\team_\ell \in \{\team_1, \ldots, \team_L\}$}
        \IndentLineComment{\textcolor{red}{\Cref{sec:hydraviper:resampling}: Dataset Resampling}}
        \State $\mathcal{D}_\ell', K_{\mathrm{drop}}' \gets$ \textbf{C-Q}$(\mathcal{D}_\ell, \team_\ell, \pi^*, Q^{\pi^*}, \epsilon)$
        \State $K_{\mathrm{drop}} \gets \min(K_{\mathrm{drop}}, K_{\mathrm{drop}}')$
        \For{each agent $i \in \team_\ell$}
            \State $\hat{\pi}_i^m \gets \textbf{TrainDT}(\mathcal{D}_\ell')$
        \EndFor
    \EndFor
\EndFor

\For{each team $\team_\ell \in \{1, \dots, L\}$}
    \IndentLineComment{\textcolor{red}{\Cref{sec:hydraviper:validation}: Validation Rollouts}}
    \State $\hat{\pi}_i, \forall i \in \team_\ell \gets$ \textbf{VR-UCB}$(\{\hat{\pi}^m_\ell\}_{m=1}^M, \team_\ell, B_{\mathrm{valid}}, c)$
\EndFor
\State \Return $\hat{\pi} = (\hat{\pi}_1, \dots, \hat{\pi}_N)$
\end{algorithmic}
\end{algorithm}

\subsection{Dataset Resampling: Centralised-Q Weighting}
\label{sec:hydraviper:resampling}
\begin{algorithm}[t]
\caption{Centralised-Q Resampling (C-Q)}
\label{alg:c-q}
\begin{algorithmic}[1]
\Statex \textbf{Input:} {\parbox[t]{0.85\linewidth}{%
    Team dataset $\mathcal{D}_\ell$, team $\team_\ell$, experts $\pi^*$, \\
    expert $Q$-functions $Q^{\pi^*}$, threshold $\epsilon$
    }
}
\Statex \textbf{Output:} {\parbox[t]{0.75\linewidth}{%
    Resampled dataset $\mathcal{D}_\ell'$, dropped rollout \\
    count $K_{\mathrm{drop}}$
    }
}

\State{\parbox[t]{0.95\linewidth}{%
    \textbf{Set} weights for each $(\jointob_k, \jointact_k) \in \mathcal{D}_\ell$: \\
    \resizebox{\linewidth}{!}{
        $p_{\ell k} \gets \bar{V}^{\pi_\ell^*}(\jointob_k) - \min_{\jointact_\ell} \bar{Q}^{\pi_\ell^*}(\jointob_k, \jointact_\ell, \pi^*_{-\ell}(\jointob_{-\ell k}))$
    }
    }
}  

\State \textbf{Update} \resizebox{0.85\linewidth}{!}{
    $K_{\mathrm{drop}} \gets \min_\ell \left\lceil \frac{1}{T} |\{(\jointob_k, \jointact_k) \in \mathcal{D}_\ell \mid p_{\ell k} \leq \epsilon\}|\right\rceil$
}

\State \textbf{Resample} dataset: $\mathcal{D}_\ell' \gets \{(\jointob_k, \jointact_k) \sim p_{\ell k}\}$
\State \Return $\mathcal{D}_\ell', K_{\mathrm{drop}}$
\end{algorithmic}
\end{algorithm}

VIPER-based algorithms include a dataset resampling step (\Cref{alg:hydraviper}, lines~9--10) so that students can focus their learning on more critical states. At a high level, they construct a training dataset by computing sample weights on the aggregated dataset of environment rollouts, typically using some notion of value based on the expert $Q$-functions. Measuring value is straightforward in the single-agent setting, but --- as we have mentioned --- a key obstacle in multi-agent learning is efficient coordination among agents. To address this challenge, HYDRAVIPER induces coordination in the resampling step using a team-based notion of value (\Cref{alg:c-q}), but trains {\dt}s independently for each agent. 

Specifically, HYDRAVIPER resamples the dataset $\mathcal{D}$ for {\dt} construction based on weights $p_{\ell k}$, which represent the relative importance of each sample for each team of agents $\team_\ell$ (\Cref{alg:c-q}, line 1). Past work computed this importance based on \emph{individual} $Q$-functions, meaning that each agent must maintain its own dataset $\mathcal{D}_i$ and induce coordination through (typically computationally expensive) joint training procedures. By contrast, we propose an intuitive change: HYDRAVIPER uses the mean of the expert $Q$-functions within each \emph{team} of coordinated agents, $\bar{Q}^{\pi_\ell^*} := \frac{1}{|\team_\ell|} \sum_{j \in \team_\ell} Q^{\pi_j^*}$, to prioritise samples according to their value to the team. Then, we compute the weights as the difference in value between the optimal joint team action and the worst-case joint team action. Intuitively, highly-weighted samples are those where coordinating on joint actions matters for performance. The weights are defined as:
\begin{align}
    p_{\ell k} &\propto \bar{Q}^{\pi_\ell^*}(\jointob_k, \pi^*(\jointob_k)) - \min_{\jointact_\ell} \bar{Q}^{\pi_\ell^*}(\jointob_k, \jointact_\ell, \pi^*_{-\ell}(\jointob_{-\ell k})) \nonumber \\
    &= \bar{V}^{\pi_\ell^*}(\jointob_k) - \min_{\jointact_\ell} \bar{Q}^{\pi_\ell^*}(\jointob_k, \jointact_\ell, \pi^*_{-\ell}(\jointob_{-\ell k})).
    \label{eq:resample_prob}
\end{align}
For further gains in sample efficiency, HYDRAVIPER does not compute $p_{\ell k}$ by enumerating joint actions over all agents in the environment. Instead, it only enumerates possible joint actions $\jointact_\ell$ over the \emph{team} and uses expert actions $\pi_{-\ell}^*(\jointob_{-\ell})$ for the opponent agents. This novel resampling procedure eliminates the need for per-agent datasets in IVIPER and MAVIPER, allowing agents to prioritise the same critical states without computationally expensive joint training. 

HYDRAVIPER uses each team's jointly sampled dataset $\mathcal{D}_\ell$ to independently train {\dt}s for each agent $i$ (\Cref{alg:hydraviper}, lines 8--9). The {\dt} $\hat{\pi}_i$ uses individual observations $o_i$ to fit $\pi_i^*$'s actions in the dataset. Modifying the input dataset rather than the training procedure provides HYDRAVIPER with flexibility in the choice of {\dt} learning algorithm. We use CART \cite{Breiman1984}, but more advanced models such as random forests or mixtures of {\dt}s \cite{Vasic2022} can also be used to improve performance.

\subsection{Training Rollouts: Adaptive Budget Allocation}
\label{sec:hydraviper:training}
\begin{algorithm}[ht]
\caption{Adaptive Training Rollouts (TR-A)}
\label{alg:tr-a}
\begin{algorithmic}[1]
\Statex \textbf{Input:} {\parbox[t]{0.85\linewidth}{%
    Dataset $\mathcal{D}$, students $\hat{\pi}^{m-1}$, iteration $m$, \\
    per-iteration rollout count $K_{\mathrm{train}}$, training \\
    rollout budget $B_{\mathrm{train}}$, dropped rollout count \\
    $K_{\mathrm{drop}}$, total rollout count $n_{\mathrm{train}}$
    }
}
\Statex \textbf{Output:} Updated dataset $\mathcal{D}$, total rollout count $n_{\mathrm{train}}$

\State \textbf{Set} $K_{\mathrm{train}}^m \gets \min(K_{\mathrm{drop}}, K_{\mathrm{train}}) \mathds{1}[n_{\mathrm{train}} \leq B_{\mathrm{train}}]$
\State \textbf{Update} $n_{\mathrm{train}} \gets n_{\mathrm{train}} + K_{\mathrm{train}}^m \mathds{1}[m > 1]$

\For{each team $\team_\ell \in \{1, \dots, L\}$}
    \State{\parbox[t]{0.85\linewidth}{%
        \textbf{Collect and relabel} 
        $K_{\mathrm{train}}^m$ rollouts: \\
        $\mathcal{D}_\ell^m \gets \{(\jointob_\ell, \pi_\ell^*(\jointob_\ell)) \sim d^{(\hat{\pi}^{m-1})}\}$
        }
    }
    \State \textbf{Aggregate} dataset: $\mathcal{D}_\ell \gets \mathcal{D}_\ell \cup \mathcal{D}_\ell^m$
\EndFor
\State \Return $\mathcal{D}, n_{\mathrm{train}}$
\end{algorithmic}
\end{algorithm}

Thus far, we have assumed that HYDRAVIPER has access to a dataset of observation-action pairs for training. To collect this dataset, HYDRAVIPER follows the DAGGER-style iterative procedure of collecting a dataset at each iteration $m$ by rolling out the current student policies $\hat{\pi}^{m-1}$ (\Cref{alg:hydraviper}, line~6--7). The next set of students are trained on the aggregate of all collected datasets, therefore building up the set of inputs likely to be encountered by the student policies during execution. However, collecting training rollouts is computationally expensive. Past work has employed an inefficient static allocation strategy that uniformly performs $K_{\textrm{train}}$ rollouts in each iteration. This strategy is problematic because the students are far from convergence early in training, so the distribution of trajectories collected earlier in training potentially diverges from those that converged students would encounter. HYDRAVIPER addresses this challenge through an adaptive rollout strategy that dynamically allocates the training budget at each iteration and prioritises critical states encountered later in training.

Recall that, for each team of cooperative agents $\team_\ell$, HYDRAVIPER follows \Cref{eq:resample_prob} to compute weights $p_{\ell k}$ for resampling the training dataset. We show the following:
\begin{theorem}
    Given a dataset of observation-action pairs for team $\team_\ell$ in iteration $m$ of HYDRAVIPER, $\mathcal{D}_\ell = \{(\jointob_\ell, \jointact_\ell)\}$, assume there exists a pair $(\jointob_{\ell k}, \jointact_{\ell k})$ that receives the weight $p_{\ell k}^{(m)} = 0$. Then, in iteration $m+1$ of HYDRAVIPER, this pair also receives the weight $p_{\ell k}^{(m+1)} = 0$.
\label{thm:rs}
\end{theorem}
\begin{proof}
    See \Cref{sec:app:rs-proof}. \phantom\qedhere
\end{proof}
As a result, samples $(\jointob_{\ell k}, \jointact_{\ell k})$ with $p_{\ell k} = 0$ are effectively \emph{removed} from the dataset $\mathcal{D}$. This intuition serves as the motivation behind HYDRAVIPER's adaptive training rollout budget allocation (\Cref{alg:tr-a}): after samples are dropped during the resampling procedure, HYDRAVIPER performs rollouts to replenish the dataset. 

Specifically, we treat the first iteration as a warm-up period, in which the experts collect a predefined number of $K_{\mathrm{train}}$ rollouts (\Cref{alg:tr-a}, lines~3--5). This leads to an initial dataset of $T \cdot K_{\mathrm{train}}$ observation-action pairs. Each team $\team_\ell$ discards non-critical samples from its dataset (\Cref{alg:c-q}, line~2), i.e. those where the range in the $Q$-value is at most a predefined threshold $\epsilon$. With the goal of efficiency in mind, HYDRAVIPER computes the minimum number of such discarded samples across all teams of cooperative agents. This then determines the minimum number of rollouts required to collect at least this many samples in the next iteration. The expected number of dropped rollouts, and therefore the budget for the next iteration, is: 
\begin{align*}
    K_{\mathrm{drop}} = \min_\ell \ceil{\frac{1}{T} \left|\{(\jointob_k, \jointact_k) \in \mathcal{D}_\ell \mid p_{\ell k} \leq \epsilon\}\right|}.
\end{align*}
During the remaining $M - 1$ iterations, HYDRAVIPER continues to collect rollouts using students until it exhausts its total budget of $B_{\mathrm{train}}$ training rollouts (\Cref{alg:tr-a}, line~1). Choosing different rollout budgets allows performance and efficiency to be traded off. A higher budget is likely to lead to superior performance, as more rollouts will be collected from students closer to convergence before the budget is exhausted, but it also requires more computation time.

\subsection{Validation Rollouts: UCB Policy Selection}
\label{sec:hydraviper:validation}
\begin{algorithm}[t]
\caption{UCB Validation Rollouts (VR-UCB)}
\label{alg:vr-ucb}
\begin{algorithmic}[1]
\Statex \textbf{Input:} {\parbox[t]{0.85\linewidth}{%
    Policies $\{\hat{\pi}^m_\ell\}_{m=1}^M$, team $\team_\ell$, validation \\
    rollout budget $B_{\mathrm{valid}}$, scaling factor $c$
    }
}
\Statex \textbf{Output:} Selected policies $\hat{\pi}_i, \forall i \in \team_\ell$

\State \textbf{Initialise} $n_m \gets 0$ for all $m \in \{1, \dots, M\}$
\State \textbf{Initialise} $n_{\min} \gets \ceil{2 \ln B_{\mathrm{valid}}}$
\State{\parbox[t]{0.85\linewidth}{%
    \textbf{Initialise} return estimates: \\
    $\mu_\ell^m \gets \frac{1}{C_{min}} \sum_{k=1}^{C_{min}} \meanreturn_{\ell k}, \meanreturn_{\ell k} \sim d^{(\hat{\pi}^m_\ell, \pi^*_{-\ell})}$
    }
}

\For{rollout $k \in \{1, \ldots, (B_{valid} - m n_{\min})\}$}
    \State \textbf{Set} $m^* \gets \argmax_m \hat{\mu}_\ell^m + \sqrt{\frac{c \ln B_{\mathrm{valid}}}{n_m}}$
    \State \textbf{Collect} mean return: $\meanreturn_{\ell k} \sim d^{(\hat{\pi}^{m^*}_\ell, \pi^*_{-\ell})}$
    \State \textbf{Update} rollout count: $n_{m^*} \gets n_{m^*} + 1$
    \State{\parbox[t]{0.85\linewidth}{%
        \textbf{Update} running average of mean return: \\
        $\hat{\mu}_\ell^{m^*} \gets \frac{n_{m^*} - 1}{n_{m^*}} \hat{\mu}_\ell^{m^*} + \frac{1}{n_{m^*}} \meanreturn_{\ell k}$
        }
    } 
\EndFor

\State \Return $\hat{\pi}_\ell^{m^*}, \forall i \in \team_\ell$
\end{algorithmic}
\end{algorithm}

Following $M$ iterations, HYDRAVIPER produces $M$ joint policy profiles for each team. It then must select the best-performing policy profile (\Cref{alg:hydraviper}, lines~13--14). HYDRAVIPER iterates through the policy profiles to estimate the team performance of each using a set of validation rollouts. The performance metric it uses is the undiscounted mean return of the team, $\meanreturn^m_\ell = \frac{1}{T} \sum_{t=0}^T \rewardfunc_\ell(s_t, \hat{\pi}^m_\ell(\jointob_t))$.

As is the case for training, collecting validation rollouts is computationally intensive, so these rollouts also need to be efficiently allocated. However, the problem setting differs here. Our goal is not to collect a \emph{diverse} set of training rollouts, but rather to identify the \emph{most performant} policy profiles using as few rollouts as possible. The mean return of each policy profile is unknown \emph{a priori}; it must be estimated by selecting policy profiles and performing rollouts with noisy returns. Again, a fixed allocation strategy of $K_{\mathrm{valid}}$ environment rollouts for each policy profile is wasteful. The rollouts assigned to clearly poorly performing policy profiles could be reallocated to reduce the variance in the estimated returns of promising policy profiles. This motivation aligns with that of multi-armed bandit (MAB) problems.

Given a limited budget of $B_{\mathrm{valid}}$ rollouts, we represent the task of selecting the best-performing policy profile as a MAB problem. For each team $\team_\ell$, the policy profile $\hat{\pi}_\ell^m$ from each iteration $m$ is an arm, and its return is a random variable $\meanreturn_\ell^m$ with unknown mean $\mu_\ell^m$. Each rollout samples from one such random variable, which captures the distribution of returns from environment and policy randomness. The objective is to identify the best arm $m_\ell^* = \argmax_m \mu_\ell^m$ in as few rollouts as possible, i.e. to minimise the regret with respect to the policy that selects $m_\ell^*$ for every rollout. 

In this work, we use a modification of the \textsc{UCB1} algorithm \cite{Auer2002}. This allows us to achieve logarithmic regret given a readily satisfiable assumption: that the returns $\meanreturn_\ell^m$ of the arms are bounded (see \Cref{sec:app:ucb}). Given a total budget of $B_{\mathrm{valid}}$ validation rollouts, HYDRAVIPER performs them as follows. For each policy profile, it first performs $n_{\min} = \ceil{2 \ln B_{\mathrm{valid}}}$ rollouts to generate initial estimates of the mean returns (\Cref{alg:vr-ucb}, lines~2--3). To allocate the remainder of the budget (lines~4--8), HYDRAVIPER follows \textsc{UCB1} to select the policy profile index for the $k$th validation rollout as
\begin{align*}
    m_{\ell k}^* = \argmax_m \left(\hat{\mu}_\ell^m(k) + \sqrt{\frac{c \ln B_{\mathrm{valid}}}{n_m(k)}}\right),
\end{align*}
where $n_m(k) = \sum_{k'=1}^k \mathds{1}[m_{\ell k'}^* = m]$ is the number of rollouts that have used policy profile $m$ thus far, $\hat{\mu}_\ell^m(k) = \frac{\sum_{k'=1}^k \meanreturn_{\ell k'} \mathds{1}[m_{\ell k'}^* = m]}{n_m(k)}$ is the empirical mean of the returns $\meanreturn_{\ell k}$ from policy profile $m$, $B_{\mathrm{valid}}$ is the total budget of rollouts, and $c$ is a scaling constant for the confidence bound (see \Cref{sec:exp:hyp}). HYDRAVIPER maintains a running average for the mean return of each policy profile, which it updates using the mean return $\meanreturn_{\ell k}$ of each rollout (line 8).

\subsection{Agent Clustering: Scaling Up HYDRAVIPER}
\label{sec:hydraviper:clustering}
When resampling the dataset, HYDRAVIPER calculates sample weights following \Cref{eq:resample_prob}. This computation requires enumerating joint actions $\jointact_\ell$ for each team $\team_\ell$, in order to find the worst-case joint action that minimises the team's mean $Q$-function, $\min_{\jointact_\ell} \bar{Q}^{\pi_\ell^*}(\jointob_k, \jointact_\ell, \pi^*_{-\ell}(\jointob_{-\ell k}))$. The complexity of this step scales with the size of the joint action space and thus exponentially with the size of the team. Some mixed competitive-cooperative environments (see \Cref{sec:exp:envs}) have an inherent team structure that can reduce this complexity. In cooperative environments such as TSC, HYDRAVIPER clusters the agent set into teams to improve training efficiency (\Cref{alg:tr-a}).

\begin{algorithm}[t]
\caption{Agent Graph Clustering (Partition)}
\label{alg:partition}
\begin{algorithmic}[1]
\Statex \textbf{Input:} {\parbox[t]{0.85\linewidth}{%
    Markov game $(\statespace, \actspace, \transitionprob, \rewardfunc_\aindex, \obsfunc_\aindex)$, experts $\pi^*$,  \\
    agent distance function $d$
    }
}
\Statex \textbf{Output:} Agent teams $\team_1, \ldots, \team_L$

\State \textbf{Construct} graph $G = (V = \{1, \ldots, N\}, E, w = 0)$

\For{each agent $i \in \{1, \dots, N\}$}
    \For{each agent $j \in \{1, \dots, N\}$}
        \State \textbf{Assign} edge weight $w_{ij} \gets \frac{1}{d(i, j)}$
    \EndFor
\EndFor

\State \textbf{Partition} graph $\team_1, \ldots, \team_L \gets \textbf{METIS}(G, L)$

\State \Return $\team_1, \ldots, \team_L$
\end{algorithmic}
\end{algorithm}

Our goal is to find a clustering of the agent set into teams $\team_1 \ldots \team_L$ so that HYDRAVIPER-trained {\dt} students have \emph{performance} similar to those trained on the full agent set, but improved \emph{scalability} in that the number of actions to enumerate per team is much smaller than the full agent set: $\prod_{i \in \team_\ell} \actspace_i \ll \prod_{i \in \{1, \ldots, N\}} \actspace_i$. We leverage the intuition that agents distant from each other (in terms of environmental distance, trajectory similarity, or other metrics) are unlikely to be influential on each other in most environments. 

Suppose that we are given a function $d(i, j)$ that computes this distance between a pair of agents. In our clustering procedure (\Cref{alg:partition}), we first construct a complete graph $G = (V, E) = K_N$ where the nodes represent agents, and the weight between node $i$ and node $j$, $w_{ij}$, is inversely proportional to $d(i, j)$ (lines 1--4). Then, we perform \emph{graph partitioning} to divide $G$ into $L$ contiguous, connected node clusters of approximately equal size (line 5), such that the sum of the weights of inter-cluster edges is minimised. We use the hierarchical METIS algorithm \cite{Karypis1998} to accomplish this. Note that we solve a graph partitioning problem instead of a min-cut problem to prevent the clusters from being imbalanced. Otherwise, in the worst case, the largest cluster could have size $O(N)$, thus yielding minimal gains in scalability.

How can the distance metric $d$ be defined? For an environment that has an inherent team structure (such as the physical deception environment in \Cref{sec:exp:envs}), we define the graph $G$ as a complete subgraph for each team. For traffic signal control environments, we note that the road network inherently forms a graph $G_{env}$, which can be partitioned to obtain sets of spatially proximal agents that correspond to neighbouring intersections. In this case, $G = G_{env}$.

Environment-agnostic distance metrics can also be designed. We consider a measure of proximity that aligns with the VIPER framework: the influence of agents on each other's $Q$-values. Recall from \Cref{sec:hydraviper:resampling} that HYDRAVIPER measures the importance of samples by the mean $Q$-function $\bar{Q}^{\pi_\ell^*}$. In a cooperative setting, if one agent's actions have a significant impact on another agent's $Q$-values, then the joint actions of these agents are likely to have a significant impact on the overall agent set's mean $Q$-function. Including both agents on the same team would allow HYDRAVIPER to capture the effects of these joint actions. We define $G$ using the distance function
\begin{align*}
    d(i, j) &= \mathbb{E}_{(\jointob, \jointact) \in \mathcal{D}_\textrm{dist}} \frac{2}{\delta_{ij} + \delta_{ji}}, \\
    \delta_{ij} &= Q^{\pi_j^*}(\jointob, \pi^*(\jointob)) - \min_{a_i} Q^{\pi_j^*}(\jointob, a_i, \pi^*_{-i}(\jointob_{-i})),
\end{align*}
where $\delta_{ij}$ is the range in agent $j$'s $Q$-values induced by agent $i$, and the distance $d_{ij}$ is the inverse of the average of $\delta_{ij}$ and $\delta_{ji}$. This symmetrisation of influence is a simplifying assumption to obtain a single weight for each edge. Alternative distance metrics could be designed to better capture agent pairs where one is significantly more influential than the other. Since METIS requires integral edge weights, we rescale $\delta_{ij}$ to percentiles between $\min_{i,j} \delta_{ij}$ and $\max_{i,j} \delta_{ij}$.

\section{Experiments}
\label{sec:exp}
Now, we demonstrate the utility of HYDRAVIPER for interpretable MARL using experiments in various benchmark environments. In doing so, we perform a functionally grounded evaluation of interpretability \cite{DoshiVelez2017}, where we assess the quality of the generated {\dt}s in terms of \emph{performance} and \emph{computational efficiency}. As the {\dt}s would be used directly in place of NN-based policies in deployment, we consider these to be good proxy metrics for their practical applicability. More specifically, we address the following research questions:
\begin{rsq}
    Is HYDRAVIPER both performant and efficient (in terms of environment interactions and runtime)?
    \label{rsq:tradeoff}
\end{rsq}
\begin{rsq}
    Does HYDRAVIPER maintain performance optimality as the environment interaction budget decreases?
    \label{rsq:budget}
\end{rsq}
\begin{rsq}
    Can HYDRAVIPER maintain performance optimality while scalability is improved through agent clustering?
    \label{rsq:scale}
\end{rsq}

\subsection{Environments}
\label{sec:exp:envs}
We evaluate HYDRAVIPER in four environments: two environments in the \emph{multi-agent particle world} (MPE) benchmark \cite{Lowe2017}, and two \emph{traffic signal control} (TSC) environments in the RESCO benchmark \cite{Ault2021}. In MPE environments, agents must navigate in a 2D space to achieve a coordinated objective, making these environments ideal for assessing coordination capabilities.

\paragraph{Cooperative navigation (CN)} In this environment, a team of three agents must coordinate to split up and cover three different targets while avoiding collisions with each other.

\paragraph{Physical deception (PD)} In this environment, a team of two defender agents must cooperate to protect two targets from an adversary agent. One of the two targets is the ``goal'' of the adversary; this is not known to the adversary, which can only observe the positions of the targets and defenders. We train the two defender agents against an NN adversary.

In TSC environments, each agent controls a single intersection by selecting different signal phases; each phase allows vehicles from a subset of lanes to pass through the intersection. Both environments are based on real-world road corridors reproduced in the traffic simulator SUMO \cite{AlvarezLopez2018}. To interface with the simulator, we use the OpenAI Gym-style wrapper \texttt{sumo-rl} \cite{Alegre2021}. We focus on imitating experts for all agents as a team.

\paragraph{Cologne corridor (CC)} \cite{Uppoor2011} This environment simulates three signalised intersections in a corridor from the city of Cologne (K\"oln), Germany. It has a total volume of 4\,494 vehicles in 7--8 am rush hour traffic. 

\paragraph{Ingolstadt corridor (IC)} \cite{Lobo2020} This larger environment simulates seven signalised intersections in a corridor from the city of Ingolstadt, Germany. It has a lower total volume of 3\,031 vehicles in 4--5 pm rush hour traffic.

\subsection{Baselines and Setup}
\label{sec:exp:baselines}
We compare HYDRAVIPER with IVIPER and MAVIPER, which represent the state of the art in interpretable multi-agent RL with {\dt} surrogate policies. In addition, we compare with \emph{expert} policies --- MADDPG \cite{Lowe2017} for MPE and MPLight \cite{Chen2020} for TSC --- and an additional baseline, \emph{imitation {\dt}}. Imitation {\dt} does not use students to collect rollouts, nor does it perform dataset resampling; it collects the same number of training rollouts as the other algorithms and trains {\dt}s on the collected dataset. As imitation {\dt} performs worse than the other algorithms by a wide margin, we do not include it in \Cref{tab:perf} or \Cref{fig:rollouts} but show its performance in \Cref{sec:app:full-perf}.

For MPE environments, we use a horizon of 25 timesteps per episode, and we trained MADDPG for 60\,000 episodes as the expert for the {\dt} students to imitate. For TSC environments, we use a horizon of 125 timesteps per episode (each timestep represents 20 seconds of simulation time), and we trained MPLight for 500 episodes as the expert. All imitation learning algorithms were run for 100 iterations to produce {\dt}s with a maximum depth of 4. IVIPER and MAVIPER ran $K_{\mathrm{train}} = K_{\mathrm{valid}} = 50$ training and validation rollouts per iteration for MPE (including for the initial iteration where rollouts are collected by the experts), and 10 rollouts per iteration for TSC. Imitation DT ran the same number of training rollouts. We set these to equalise the number of environment interactions per iteration.

We repeated all experiments 10 times with different random seeds, and we report the mean and 95\% confidence interval of the reward over 10 rollouts performed with the final student policy profiles generated from these runs. Most experiments were run in parallel on a server with 56 2.75GHz AMD EPYC 7453 processors and 252 GiB of RAM. For these experiments, we report the \emph{number of rollouts collected}, not \emph{runtime}s, as the rollout time is roughly constant. However, we also report runtimes for the execution of IVIPER, MAVIPER, and HYDRAVIPER on all four environments. For these experiments, we use the \texttt{kernprof} profiler (v4.1.3) to run them in sequence, with no other concurrent processes running except system routines. These experiments were run on another server with 8 4.2GHz Intel i7-7700K processors and 62 GiB of RAM. 

\subsection{Results}
\label{sec:exp:results}
\begin{table*}[t]
    \centering
    \resizebox{\linewidth}{!}{
    \begin{tabularx}{1.075\linewidth}{XXccccc}
        \toprule
        \multicolumn{2}{l}{\textbf{Environment}}                                 & \textbf{Expert}   & \textbf{IVIPER}       & \textbf{MAVIPER}      & \textbf{HYDRAVIPER}   & \textbf{HYDRAVIPER LB} \\
        \midrule
        \textbf{Cooperative Navigation}      & \textit{Total Penalty}            & 122.67 $\pm$ 1.67 &    160.87 $\pm$ 4.31  &    144.35 $\pm$ 2.12  &    144.48 $\pm$ 2.67  & 144.84 $\pm$ 2.12        \\
                                             & \textit{Runtime (s)}              &         N/A       &  2\,444.6 $\pm$ 9.1  & 21\,188.7 $\pm$ 408.6 &    206.2 $\pm$ 11.1  &  180.5 $\pm$ 9.3         \\
        \midrule                                                                                                                                                            
        \textbf{Physical \newline Deception} & \textit{Defender \newline Reward} &   8.19 $\pm$ 0.50 &      6.94 $\pm$ 0.52  &      7.74 $\pm$ 0.82  &      7.72 $\pm$ 0.53  &   7.12 $\pm$ 0.84        \\
                                             & \textit{Runtime (s)}              &         N/A       &  2\,017.2 $\pm$ 21.3  & 11\,782.4 $\pm$ 137.8 &    1\,173.5 $\pm$ 21.6  &  388.4 $\pm$ 5.6         \\
        \midrule
        \textbf{Cologne Corridor}            & \textit{Queue Length}             &  18.94 $\pm$ 2.49 &     22.06 $\pm$ 2.91  &     25.85 $\pm$ 5.22  &     16.72 $\pm$ 1.51  &  18.77 $\pm$ 3.69        \\
                                             & \textit{Runtime (s)}              &         N/A       & 33\,841.4 $\pm$ 441.3 & 37\,503.8 $\pm$ 834.8 & 13\,651.6 $\pm$ 254.2 & 1\,865.4 $\pm$ 26.1        \\
        \midrule
        \textbf{Ingolstadt Corridor}         & \textit{Queue Length}             &  23.01 $\pm$ 1.10 &     21.51 $\pm$ 2.13  &     24.26 $\pm$ 2.54  &     19.77 $\pm$ 1.51  &  21.87 $\pm$ 1.59        \\
                                             & \textit{Runtime (s)}              &         N/A       & 75\,927.5 $\pm$ 203.3 & 58\,676.4 $\pm$ 1\,392.3 & 11\,263.8 $\pm$ 55.8 & 6\,462.4 $\pm$ 30.9        \\
        \bottomrule
    \end{tabularx}
    }
    
    \caption{Performance and runtimes (means and 95\% confidence intervals) for HYDRAVIPER and baselines. All algorithms are given the same environment interaction budget, except for low budget (LB) HYDRAVIPER (which uses 20\% of the rollouts for MPE, 10\% of the rollouts for TSC). HYDRAVIPER achieves or exceeds the performance of MAVIPER using a fraction of the runtime, and still performs well in the low budget setting. For physical deception, higher rewards are better; for all other environments, lower rewards are better. \Cref{sec:app:full-rt} shows runtimes for individual algorithm steps.}
    \label{tab:perf}
\end{table*}

\paragraph{\Cref{rsq:tradeoff}} \textbf{HYDRAVIPER achieves strong, coordinated performance in a computationally efficient manner.}
First, we assess HYDRAVIPER's performance as we vary it between two environment interaction budget levels, high (5\,000 training/5\,000 validation rollouts for MPE, 1\,000 training/1\,000 validation rollouts for TSC) and low (500 training/1\,500 validation rollouts for MPE, 100 training/100 validation rollouts for TSC). As shown in \Cref{tab:perf}, HYDRAVIPER students perform better than or comparable to students trained by the most performant {\dt} baseline (MAVIPER for MPE, IVIPER for TSC) in all environments at both budget levels. HYDRAVIPER's performance is also better than or comparable to the NN experts for all environments except cooperative navigation, in which all DT-based algorithms cannot achieve expert-level performance. In physical deception, although neither MAVIPER nor HYDRAVIPER substantially outperforms IVIPER given the considerable stochasticity in the environment, HYDRAVIPER achieves a level of performance much closer to MAVIPER, while its training time is an order of magnitude shorter than MAVIPER. 

In TSC environments, HYDRAVIPER is the best performing algorithm at both the high and low interaction budget levels. Notably, HYDRAVIPER at the high budget level substantially outperforms the expert on the Ingolstadt corridor. By contrast, MAVIPER fails to coordinate the intersection agents and is in general the worst-performing algorithm. HYDRAVIPER more than halves the runtime of both IVIPER and MAVIPER on both TSC environments. 

In \Cref{fig:example-dt}, we show a decision tree generated by HYDRAVIPER for one agent in the Ingolstadt corridor at the low budget level (100 training/100 validation rollouts). Even with this limited environment interaction budget, HYDRAVIPER generates {\dt}s that are not only performant, but also intuitively sensible. Each internal node of the {\dt} compares the number of queueing vehicles for some turning movement to a threshold; the left branch includes all samples where there are fewer vehicles than the threshold, while the right branch includes all samples where there are more than the threshold. The root assesses traffic from the westbound side street. If it is low, the {\dt} coordinates north-south traffic on the main road; if it is high, the {\dt} coordinates turn traffic from the side street. Stakeholders, such as traffic engineers, can follow such a workflow to visualise, reason about, and supervise {\dt}s generated by HYDRAVIPER. 

\begin{figure*}[ht]
    \centering
    \includegraphics[trim={0cm 1cm 1cm 0cm},clip,width=0.685\textwidth]{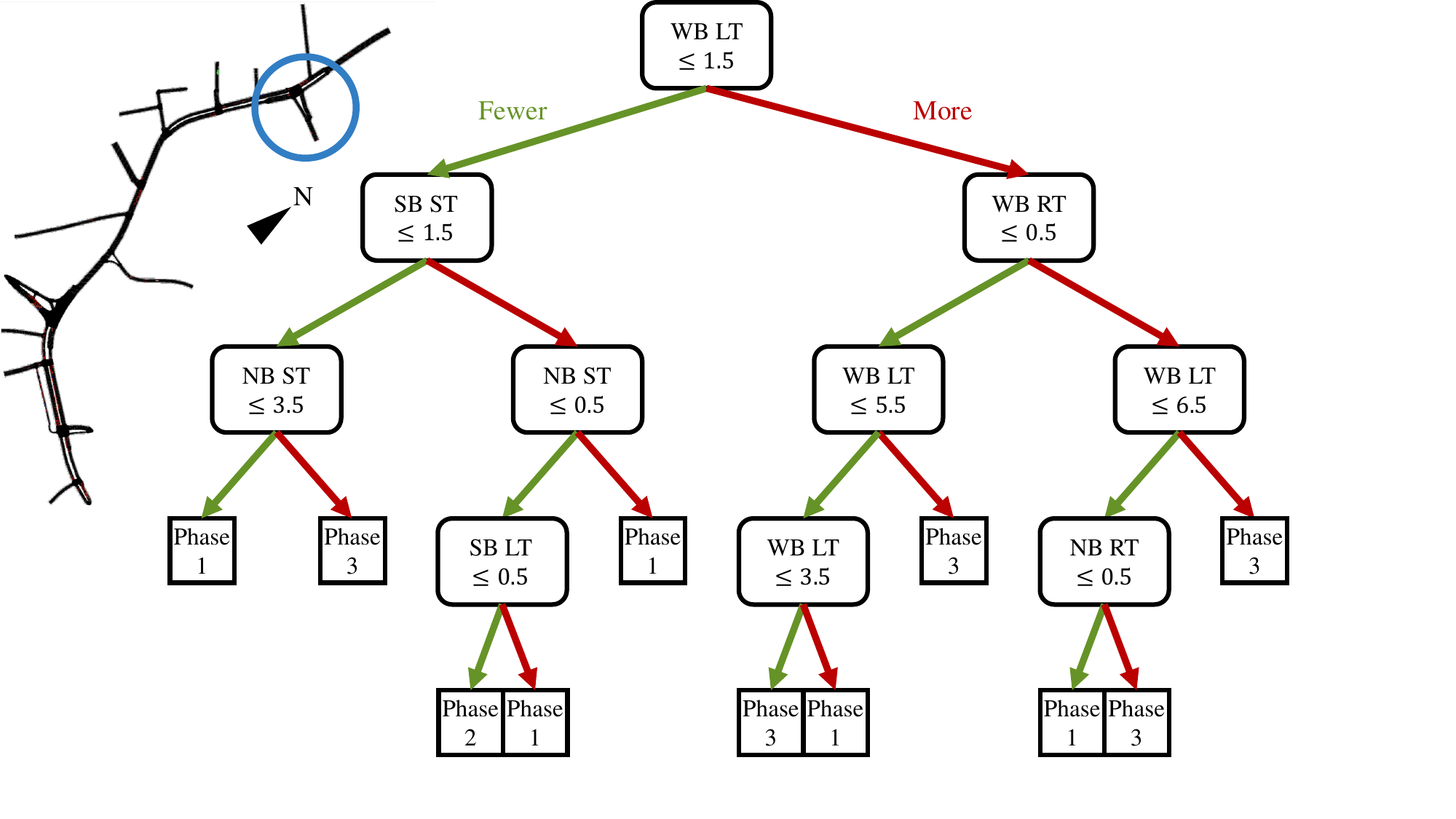}
    \caption{Decision tree generated by HYDRAVIPER (100 training/100 validation rollouts) for the Ingolstadt corridor (IC).}
    \label{fig:example-dt}
\end{figure*}

\begin{figure*}[ht]
    \centering
    \includegraphics[width=\textwidth]{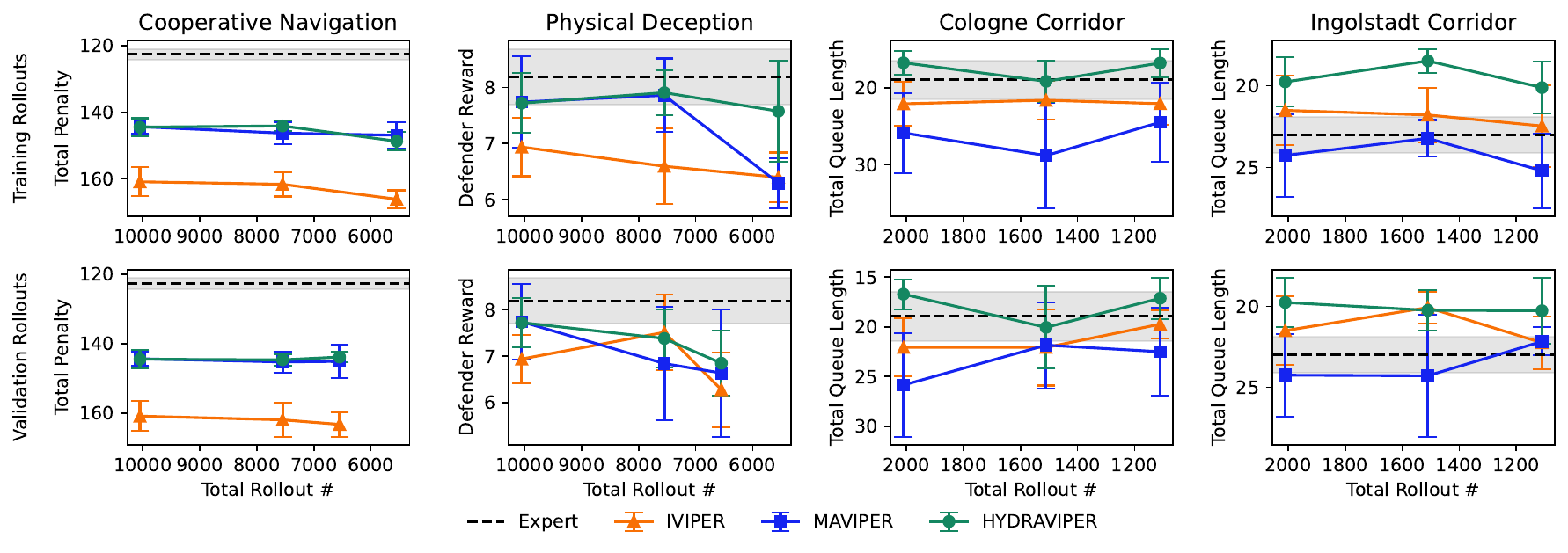}
    \caption{Performance of HYDRAVIPER and baselines as the number of rollouts decreases. Top shows decreasing training rollouts; bottom shows decreasing validation rollouts. HYDRAVIPER's performance stays consistent as the number of rollouts decreases. For physical deception, higher rewards are better; for all other environments, lower rewards are better. Bars show 95\% confidence intervals based on 10 randomly-seeded runs. Full results are shown in \Cref{tab:full-perf} in \Cref{sec:app:full-perf}.
    }
    \label{fig:rollouts}
\end{figure*}

\paragraph{\Cref{rsq:budget}} \textbf{As the environment interaction budget decreases, HYDRAVIPER still outperforms baselines.}
Now, we investigate the ability of HYDRAVIPER to adapt to increasing budget constraints for environment interaction, as would be imposed by users who wish to quickly iterate on {\dt} policy training. As shown in \Cref{fig:rollouts}, HYDRAVIPER's performance in all four environments does not change substantially as the training and validation rollout budgets are individually reduced. Furthermore, in all four environments, HYDRAVIPER achieves performance on par with or better than MAVIPER at all budget levels. Therefore, HYDRAVIPER is able to maintain a Pareto frontier in the tradeoff between performance and computational efficiency.

In cooperative navigation, the performance of both HYDRAVIPER and MAVIPER remains similar as the training and validation budgets are reduced individually. However, when both budgets are reduced simultaneously (shown in \Cref{fig:ablation-cn-cc}), the performance of HYDRAVIPER but not MAVIPER remains essentially unchanged. In physical deception, HYDRAVIPER still performs well even as its training budget is reduced by a factor of 10, whereas MAVIPER performs substantially worse. Furthermore, the 95\% confidence intervals of HYDRAVIPER's rewards are smaller than those of MAVIPER at all validation budget levels. Thus, HYDRAVIPER is able to identify performant policy profiles more consistently than MAVIPER.

In the Cologne corridor, HYDRAVIPER's performance consistently remains within the expert's 95\% confidence interval at all environment interaction budget levels, whereas the same is not true of MAVIPER. Meanwhile, the performance of HYDRAVIPER on the Ingolstadt corridor substantially exceeds the expert at all budget levels, whereas MAVIPER and IVIPER (except for the 500 validation rollout setting) remain in the expert's 95\% confidence interval.

\paragraph{\Cref{rsq:scale}} \textbf{Even when the agent set is decomposed through clustering, HYDRAVIPER maintains its performance.}
Finally, we evaluate the effect of agent clustering on the performance of HYDRAVIPER. For the Ingolstadt corridor environment in the high budget setting (1\,000 training/1\,000 validation rollouts), we evaluate two strategies from \Cref{sec:hydraviper:clustering}: (1) clustering the agent set into two teams based on the road network graph $G_{env}$ (\texttt{graph-metis}), and (2) using pairwise $Q$ values to identify mutually impactful agents, and either $k$-means clustering (\texttt{marginal-kmeans}) or METIS (\texttt{marginal-metis}) for partitioning. 

As shown in \Cref{fig:split}, these clustering strategies allow HYDRAVIPER to retain its performance even when the size of the agent set is approximately halved for each team. Clustering reduces the largest team's joint action set in size from 1\,944 to 72, and the total runtime of HYDRAVIPER by up to 47\%. The best-performing strategy combines pairwise $Q$-value weights with METIS for partitioning, but they all perform similarly to the unpartitioned algorithm.

These clustering methods also outperform two baselines. First, the \texttt{random} baseline randomly assigns each agent to one of two teams; this baseline has very high variance in performance. Second, the \texttt{contiguous} baseline uses a handcrafted division of the agent set into two subsets; \texttt{graph-metis} recovers this division automatically, but \texttt{marginal-metis} improves further by grouping agents that are not adjacent in the environment.

\begin{figure*}[ht]
    \centering
    \small
    \begin{subfigure}[ht]{.4\linewidth}
    \includegraphics[width=\linewidth]{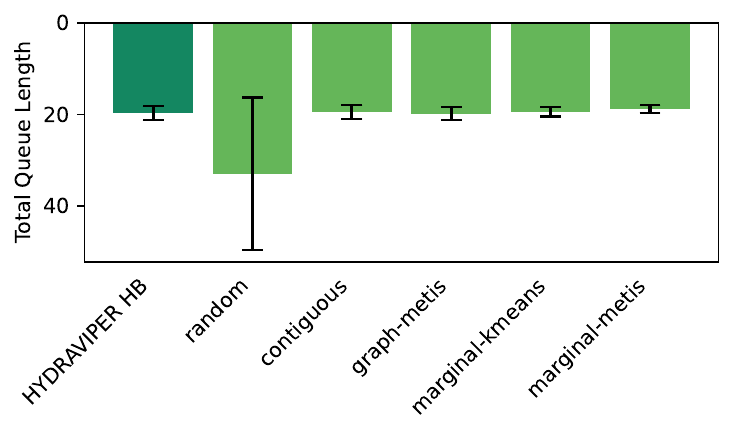}
    \end{subfigure}
    \begin{subfigure}[ht]{.575\linewidth}
    \centering
    \begin{tabularx}{0.8\linewidth}{llX}
        \toprule
         \textbf{Clustering Method} & \textbf{Runtime (s)} & \textbf{Worst \newline Clustering} \\
         \midrule
         HYDRAVIPER HB & 11\,263.8 $\pm$ 55.8 & [[1,2,3,4,5,6,7]] \\
         \midrule
         + random & 6\,925.5 $\pm$ 897.0 & [[1,2,3,6,7],[4,5]] \\
         \midrule
         + contiguous & 6\,731.3 $\pm$ 214.3 & [[1,2,3],[4,5,6,7]] \\
         + graph-metis & 6\,026.4 $\pm$ 39.1 & [[1,2,3],[4,5,6,7]] \\
         \midrule
         + marginal-kmeans & 8\,538.0 $\pm$ 643.0 & [[1,2,3,4,6,7],[5]] \\
         + marginal-metis & 6\,067.5 $\pm$ 57.0 & [[1,2,4,7],[3,5,6]] \\
         \bottomrule
    \end{tabularx}
    \end{subfigure}
    
    \caption{(L) Performance of HYDRAVIPER on the Ingolstadt corridor (IC) under different agent set clustering methods at the high budget (HB) level. (R) Runtimes and worst-performing clusterings across 10 different random seeds of HYDRAVIPER under these methods. The intersection agent numbers follow those shown in \Cref{fig:algorithm}.}
    \label{fig:split}
\end{figure*}

\subsection{Hyperparameter Sensitivity}
\label{sec:exp:hyp}
To understand the effects of HYDRAVIPER's hyperparameters on its performance, we conduct experiments to vary the depth of the {\dt} students, and the scaling constant $c$ for UCB policy selection (\Cref{sec:hydraviper:validation}), on the cooperative navigation environment. We choose this environment due to its relatively low level of randomness. For these experiments, we use HYDRAVIPER at the low budget level (500 training/1\,500 validation rollouts) as the baseline algorithm, and fix all hyperparameters other than those of interest. Figures for all results are shown in \Cref{sec:app:full-hyp}.

By default, we use a {\dt} depth of 4; our results show that {\dt}s of this depth provide a good tradeoff between expressiveness and computational efficiency. Depth-4 {\dt}s outperform depth-2 and depth-3 {\dt}s on cooperative navigation. This is an intuitive result; the optimal agent policy for this environment cannot be represented with such shallow {\dt}s, as they must condition on the positions of the other agents and the landmarks. However, we find that depth-4 {\dt}s also marginally outperform depth-5 {\dt}s. This same pattern exists in all of the environments that we use for evaluation. We hypothesise that the amount of data collected by HYDRAVIPER at the low budget level is insufficient to coordinate between depth-5 {\dt}s.

Our default value for $c$ is also 4. As discussed in \Cref{sec:app:ucb}, this value is smaller than would be necessary according to a theoretical analysis. Since the agents are already trained, we hypothesise that the potential range of returns is less useful in practice for finding a good policy profile than the typical range of returns. The best {\dt} depths for the other environments are all 4, as with cooperative navigation, while the best tested values of $c$ for physical deception and the Cologne corridor are, respectively, 2 and 16.

\begin{figure*}[t]
    \centering
    \begin{subfigure}[ht]{.45\linewidth}
    \includegraphics[width=\linewidth]{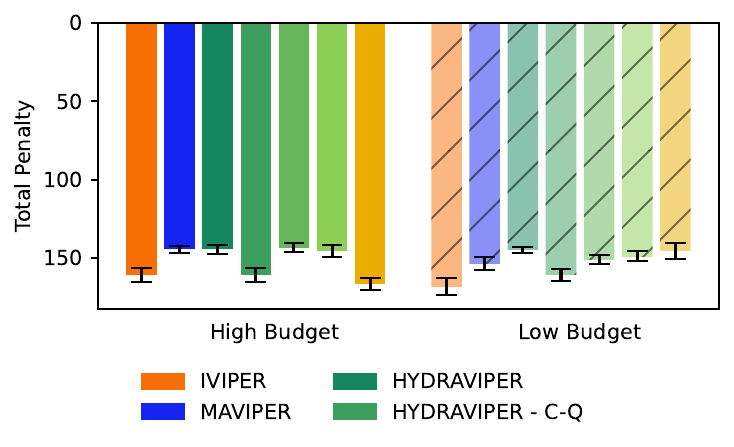}
    \end{subfigure}
    \begin{subfigure}[ht]{.45\linewidth}
    \includegraphics[width=\linewidth]{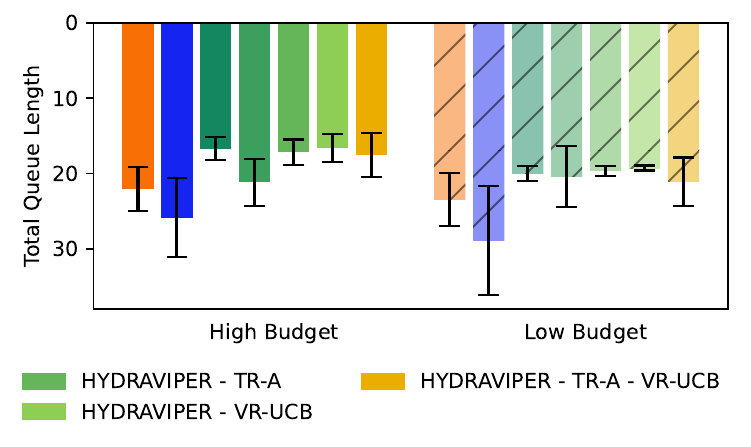}
    \end{subfigure}
    \caption{Ablation on cooperative navigation (CN, L) and Cologne corridor (CC, R). Lower rewards are better.}
    \label{fig:ablation-cn-cc}
\end{figure*}

\subsection{Ablation}
\label{sec:exp:ablation}
Lastly, to understand which components of HYDRAVIPER are responsible for its success, we conduct an ablation study for HYDRAVIPER at two budget levels in the cooperative navigation (CN) and Cologne corridor (CC) environments. For the high budget level, we use 5\,000 training/5\,000 validation rollouts for CN, and 1\,000 training/1\,000 validation rollouts for CC; for the low budget level, we use 500 training/1\,500 validation rollouts for CN, and 100 training/100 validation rollouts for CC. We compare HYDRAVIPER's centralised-$Q$ resampling with IVIPER's independent resampling (HYDRAVIPER - CQ). In addition, we study the impact of removing adaptive training budget allocation (HYDRAVIPER - TR-A) and UCB-based validation budget allocation (HYDRAVIPER - VR-UCB). Results for physical deception and the Ingolstadt corridor are shown in \Cref{sec:app:full-ablation}.

\Cref{fig:ablation-cn-cc} shows our ablation results. In both environments and at both budget levels, centralised-$Q$ resampling outperforms the IVIPER resampling scheme, although the performance gap is less pronounced for the Cologne corridor due to environmental randomness. This result suggests that sampling the training dataset independently for each agent, instead of according to team performance, is insufficient to achieve coordinated behaviour in the resulting students. Meanwhile, removing the budget allocation methods degrades the performance of HYDRAVIPER. Having either one of the budget allocation methods is generally sufficient to improve HYDRAVIPER's reward, except in one case: for cooperative navigation at the low budget level, HYDRAVIPER performs worse when only one budget allocation mechanism is present. Meanwhile, for the Cologne corridor at the low budget level, the variance in HYDRAVIPER's reward is large both when \emph{only} centralised-$Q$ resampling is present, and also when it is \emph{removed}. These results suggest that the primary benefit of the two rollout budget allocation mechanisms is to stabilise HYDRAVIPER's learning process, especially in the low budget setting when extracting the most information from each rollout is critical.

\section{Conclusion and Future Work}
\label{sec:conclusion}
In this work, we introduced a new {\dt}-based interpretable MARL method, HYDRAVIPER. HYDRAVIPER addresses several limitations of prior multi-agent methods that follow the VIPER framework: (1) it improves performance by using a joint dataset resampling scheme based on team $Q$-values, and (2) it improves computational efficiency by adaptively allocating fixed budgets of environment interactions for training and validation, as well as by dividing agents into jointly-trained teams. Based on experiments in benchmark environments for multi-agent coordination and traffic signal control, we showed that HYDRAVIPER achieves performance comparable with MAVIPER (a centralised method) and even neural network experts, all with a runtime less than IVIPER (a decentralised method). We also demonstrated HYDRAVIPER's sample efficiency in its ability to retain a similar level of performance using a fraction of the environment interactions.

Through our experiments in the Ingolstadt corridor environment, we scaled up the VIPER framework to seven agents. To our knowledge, this is the largest team of coordinated agents to which interpretable MARL has been applied so far. However, environments based on real-world domains can have many more agents than the environments that we studied. For example, the review of \citet{Noaeen2022} showed that TSC environments of dozens or even hundreds of agents are used in the RL literature. In the most extreme case, \citet{Chen2020} used parameter-shared MPLight policies as controller agents for an extremely large simulation of 2\,510 traffic lights. Our agent clustering approach shows promise in scaling up to larger environments while retaining performance comparable to that of expert policies. We envision that the flexibility of the HYDRAVIPER framework will allow it to adapt to characteristics of different MARL environments while maintaining Pareto optimality in the performance-computational efficiency tradeoff.

However, our functionally grounded evaluation has not shown whether the resulting {\dt}s are sufficient to help stakeholders better understand these policies. An application-grounded evaluation \cite{DoshiVelez2017} of HYDRAVIPER that includes user studies would be necessary to assess its practical utility. Combining the general framework of HYDRAVIPER with alternative policy structures, including those based on natural language, may help make these {\dt}-based policies more understandable and controllable.


\section*{Acknowledgments}
This work was supported by the Tang Family Endowed Innovation Fund and NSF grant IIS-2046640 (CAREER). Additionally, we thank Naveen Raman, Yixuan Xu, Jingwu Tang, Matteo Pozzi, and Peter Stone for their feedback.

\FloatBarrier
\newpage
\bibliography{ref}

\newpage
\appendix
\onecolumn
\section{Proof of Resampling Theorem}
\label{sec:app:rs-proof}

Here, we provide a proof for \Cref{thm:rs} in \Cref{sec:hydraviper:training}.

\setcounter{theorem}{0}
\begin{theorem}
    Given a dataset of observation-action pairs for team $\team_\ell$ in iteration $m$ of HYDRAVIPER, $\mathcal{D}_\ell = \{(\jointob_\ell, \jointact_\ell)\}$, assume there exists a pair $(\jointob_{\ell k}, \jointact_{\ell k})$ that receives the weight $p_{\ell k}^{(m)} = 0$. Then, in iteration $m+1$ of HYDRAVIPER, this pair also receives the weight $p_{\ell k}^{(m+1)} = 0$.
\end{theorem}

\begin{proof}
    If $p_{\ell k}^{(m)} \propto \bar{V}^{\pi_\ell^*}(\jointob_k) - \min_{\jointact_\ell} \bar{Q}^{\pi_\ell^*}(\jointob_k, \jointact_\ell, \pi^*_{-\ell}(\jointob_{-\ell k})) = 0$, then by definition
    \begin{align}
    \begin{split}
        \bar{V}^{\pi_\ell^*}(\jointob_k) := \max_{\jointact_\ell} \bar{Q}^{\pi_\ell^*}(\jointob_k, \jointact_\ell, \pi^*_{-\ell}(\jointob_{-\ell k})) 
        = \min_{\jointact_\ell} \bar{Q}^{\pi_\ell^*}(\jointob_k, \jointact_\ell, \pi^*_{-\ell}(\jointob_{-\ell k})),
    \label{eq:p-iteration-m}
    \end{split}
    \end{align}
    i.e. joint team actions $\jointact_\ell$ have no effect on the value given observation $\jointob_k$. When HYDRAVIPER resamples the dataset in iteration $m$ (\Cref{alg:c-q}, line~3), $(\jointob_{\ell k}, \jointact_{\ell k})$ will not be part of the resampled dataset $\mathcal{D}_\ell'$. However, the resampled dataset $\mathcal{D}_\ell'$ does not replace the original dataset $\mathcal{D}_\ell$.
    
    In iteration $m + 1$, $\mathcal{D}_\ell$ is aggregated with a newly-collected dataset of observation-action pairs $\mathcal{D}_\ell^m$ (\Cref{alg:tr-a}, line~5), and $(\jointob_{\ell k}, \jointact_{\ell k})$ continues to be part of this dataset. A new set of weights $p_\ell^{(m+1)}$ are computed using this expanded dataset (\Cref{alg:hydraviper}, line~8). Assume that $p_{\ell k}^{(m+1)} \ne 0$. Without loss of generality, let $p_{\ell k}^{(m+1)} > 0$. Then by definition
    \begin{align*}
        \bar{V}^{\pi_\ell^*}(\jointob_k) := \max_{\jointact_\ell} \bar{Q}^{\pi_\ell^*}(\jointob_k, \jointact_\ell, \pi^*_{-\ell}(\jointob_{-\ell k}))  
        > \min_{\jointact_\ell} \bar{Q}^{\pi_\ell^*}(\jointob_k, \jointact_\ell, \pi^*_{-\ell}(\jointob_{-\ell k})).
    \end{align*}
    None of $\jointob_k$, $\pi^*_{-\ell}$, $\bar{V}^{\pi_\ell^*}$, or $\bar{Q}^{\pi_\ell^*}$ changed between iterations $m$ and $m+1$, since HYDRAVIPER uses experts and expert value functions to compute the weights. This contradicts \Cref{eq:p-iteration-m}. Thus, if $p_{\ell k}^{(m)} = 0$, then $p_{\ell k}^{(m+1)} = 0$.
\end{proof}

A similar proof holds if the strict equality is replaced by the threshold $\epsilon$.

\section{Application of UCB1 to HYDRAVIPER}
\label{sec:app:ucb}
Although UCB assumes that the arms are bounded in $[0, 1]$, it can be modified in a manner equivalent to rescaling the rewards to remain in $[0, 1]$. HYDRAVIPER relies on the general form of the Chernoff-Hoeffding bound:
\begin{theorem}
    (Theorem 2 of \citet{Hoeffding1963}) For independent random variables $X_1, \ldots, X_n$ with mean $\mu$, $\bar{X} = \frac{1}{n} \sum_{i=1}^n X_i$, and $a_i \leq X_i \leq b_i, \forall i \in \{1, \ldots, n\}$, then for $\alpha > 0$
    \begin{align*}
        \Pr(\bar{X} \geq \mu + \alpha) \leq e^{-\frac{2n^2 \alpha^2}{\sum_{i=1}^n (b_i - a_i)^2}}.
    \end{align*}
    \label{thm:hoeffding}
\end{theorem}
\begin{corollary}
    Assume that $\meanreturn_\ell^m$ is bounded by $[a, b]$ with $\Delta = b - a$ for all $i, m$. For $c = 2\Delta^2$, \Cref{thm:hoeffding} shows that, for the empirical mean $\hat{\mu}_\ell^m$, (notation simplified for clarity)
    \begin{align*}
        \Pr\left(\hat{\mu}_\ell^m \geq \mu_\ell^m + \sqrt{\frac{c \ln B}{n_m}}\right) \leq e^{-\frac{2c \ln B}{\Delta^2}} = B^{-4}.
    \end{align*}
    \label{cor:c}
\end{corollary}
This is the same bound as demonstrated for \textsc{UCB1}, and the same holds for the lower confidence bound $\mu_\ell^m - \sqrt{\frac{c \ln B}{n_m}}$. Overall, this choice of $c$ leads to the same $O(\log B)$ regret bound as UCB1. However, HYDRAVIPER can also be extended to use other MAB algorithms. If the mean returns of each policy profile are assumed to be normally distributed, the \textsc{UCB1-Normal} algorithm \cite{Auer2002} could be used; it also achieves logarithmic regret. This algorithm effectively chooses $c$ to be proportional to the arms' sample variance; the greater the variance, the wider the confidence bound. An offline (but biased) estimate of the sample variance can also be obtained by performing expert rollouts before running UCB.

How do these theoretical results apply empirically? In the cooperative navigation environment, the agents navigate in a $2 \times 2$ square environment, and thus the maximum distance of an agent to a target is $2 \sqrt{2}$. The reward in this environment is the negation of the minimum agent distance to each landmark, plus a penalty of -1 for each agent that the ego agent collides with. Therefore, the maximum possible penalty is $\Delta = 3 \cdot 2\sqrt{2} + 2$, which requires $c \approx 219.88$ to achieve the guarantee of \Cref{cor:c}. However, as shown on the right of \Cref{fig:hyp-cn}, $c = 4$ empirically performs the best; $c = 8$ and $c = 16$ are already excessively conservative given the low randomness in the environment. In physical deception, the reward is the $L_2$ distance of the adversary to the goal, minus the minimum $L_2$ distance of any defender to the goal. Again, the maximum distance from an agent to the goal is $\Delta = 2 \sqrt{2}$. Therefore, $c = 16$ is required to achieve the guarantee of \Cref{cor:c}. In practice, $c = 2$ still led to good performance. In the TSC environments, the reward in these environments is the total queue length (i.e. the total number of stopped vehicles in all lanes), which is effectively unbounded, but the best values empirically are $c = 16$ and $c = 4$.

\section{Full Performance Results}
\label{sec:app:full-perf}
In \Cref{tab:full-perf}, we show the mean and 95\% confidence intervals of the performance of IVIPER, MAVIPER, HYDRAVIPER, the expert, and imitation DT, on all four environments with all of the training and validation rollout budgets we tested.

\begin{table*}[ht]
    \centering
    \resizebox{\linewidth}{!}{
    \begin{tabularx}{1.075\linewidth}{Xrrccccc}
        \toprule
        \textbf{Environment} &  \textbf{Training} & \textbf{Validation} & \textbf{Expert}   & \textbf{Imitation DT} & \textbf{IVIPER}   & \textbf{MAVIPER}  & \textbf{HYDRAVIPER} \\
        \midrule	                         
        \textbf{Cooperative} &			     5000 & 			   5000 & 122.67 $\pm$ 1.67 &	221.19 $\pm$  8.58   & 160.87 $\pm$ 4.31 & 144.35 $\pm$ 2.12 & 144.48 $\pm$ 2.67   \\
		\textbf{Navigation}	 &			     2500 & 			   5000 &                   &	211.84 $\pm$  6.32   & 161.62 $\pm$ 3.66 & 146.28 $\pm$ 3.34 & 144.13 $\pm$ 1.59   \\
							 &			      500 & 			   5000 &                   &	218.22 $\pm$  5.90   & 166.11 $\pm$ 2.67 & 146.91 $\pm$ 3.90 & 148.71 $\pm$ 2.80   \\
							 &			     5000 & 			   2500 &                   &				  	     & 161.94 $\pm$ 4.93 & 145.27 $\pm$ 3.07 & 144.66 $\pm$ 1.62   \\
							 &			     5000 & 			   1500 &                   &				  	     & 163.25 $\pm$ 3.60 & 145.10 $\pm$ 4.69 & 143.86 $\pm$ 1.54   \\
							 &			      500 & 			   1500 &                   &				  	     & 168.31 $\pm$ 5.52 & 153.66 $\pm$ 4.07 & 144.84 $\pm$ 2.12   \\
        \midrule			                                  	   		 						                           
        \textbf{Physical}    &			     5000 & 			   5000 &   8.19 $\pm$ 0.50 &	  6.27 $\pm$  0.43	 &   6.94 $\pm$ 0.52 &   7.74 $\pm$ 0.82 &   7.72 $\pm$ 0.53   \\
		\textbf{Deception}	 &		         2500 & 			   5000 &                   &	  5.73 $\pm$  0.31	 &   6.60 $\pm$ 0.67 &   7.84 $\pm$ 0.66 &   7.91 $\pm$ 0.40   \\
							 &		          500 & 			   5000 &                   &	  5.32 $\pm$  0.61	 &   6.40 $\pm$ 0.44 &   6.30 $\pm$ 0.45 &   7.58 $\pm$ 0.91  \\
							 &		         5000 & 			   2500 &                   &		                 &   7.51 $\pm$ 0.81 &   6.84 $\pm$ 1.21 &   7.38 $\pm$ 0.63  \\
							 &		         5000 & 			   1500 &                   &						 &   6.28 $\pm$ 0.80 &   6.64 $\pm$ 1.37 &   6.85 $\pm$ 0.70  \\
							 &		          500 & 			   1500 &                   &						 &   6.03 $\pm$ 0.70 &   7.36 $\pm$ 0.99 &   7.12 $\pm$ 0.84   \\
        \midrule	                           			           			      						                     
        \textbf{Cologne}    & 		         1000 & 			   1000 &  18.94 $\pm$ 2.49 &   137.67 $\pm$  0.64   &  22.06 $\pm$ 2.91 &  25.85 $\pm$ 5.22 &  16.72 $\pm$ 1.51   \\
		\textbf{Corridor}	& 		          500 & 			   1000 &                   &   211.84 $\pm$  8.00   &  21.60 $\pm$ 2.50 &  28.82 $\pm$ 6.86 &  19.13 $\pm$ 2.70   \\
							& 		          100 & 			   1000 &                   &   218.22 $\pm$ 10.29   &  22.05 $\pm$ 2.72 &  24.47 $\pm$ 5.21 &  16.75 $\pm$ 1.85   \\
							& 		         1000 & 			    500 &                   &						 &  22.07 $\pm$ 3.83 &  21.84 $\pm$ 4.33 &  20.06 $\pm$ 4.15   \\
							& 		         1000 & 			    100 &                   &						 &  19.73 $\pm$ 1.43 &  22.50 $\pm$ 4.38 &  17.12 $\pm$ 2.07   \\
							& 		          100 & 			    100 &                   &						 &  23.40 $\pm$ 3.53 &  28.91 $\pm$ 7.27 &  18.77 $\pm$ 3.69   \\
        \midrule	                           			           			      						                      
        \textbf{Ingolstadt} & 		         1000 & 			   1000 &  23.01 $\pm$ 1.10 &	169.43 $\pm$  5.26	 &  21.51 $\pm$ 2.91 &  24.26 $\pm$ 2.54 &  19.77 $\pm$ 1.51   \\
		\textbf{Corridor}	& 		          500 & 			   1000 &                   &	170.55 $\pm$  3.41	 &  21.79 $\pm$ 1.67 &  23.21 $\pm$ 1.12 &  18.48 $\pm$ 0.74   \\
							& 		          100 & 			   1000 &                   &	166.99 $\pm$  3.67	 &  22.46 $\pm$ 2.52 &  25.20 $\pm$ 2.29 &  20.11 $\pm$ 1.57   \\
							& 		         1000 & 			    500 &                   &						 &  20.08 $\pm$ 0.97 &  24.30 $\pm$ 3.79 &  20.25 $\pm$ 1.26   \\
							& 		         1000 & 			    100 &                   &						 &  22.26 $\pm$ 1.62 &  22.18 $\pm$ 0.87 &  20.28 $\pm$ 2.04   \\
							& 		          100 & 			    100 &                   &						 &  27.23 $\pm$ 5.31 &  23.31 $\pm$ 1.44 &  21.87 $\pm$ 1.59   \\
        \bottomrule
    \end{tabularx}
    }
    
    \caption{Performance (means and 95\% confidence intervals) for IVIPER, MAVIPER, and HYDRAVIPER at different training and validation budget levels. For physical deception, higher rewards are better; for other environments, lower rewards are better.}
    \label{tab:full-perf}
\end{table*}

\FloatBarrier
\section{Full Runtime Results}
\label{sec:app:full-rt}
In \Cref{fig:full-rt-mpe} and \Cref{fig:full-rt-tsc}, we show the mean and 95\% confidence intervals of the runtimes of IVIPER, MAVIPER, and HYDRAVIPER on the four most costly algorithm steps: training rollouts, dataset resampling, DT training, and validation rollouts. \Cref{fig:full-rt-mpe} shows the runtimes for the multi-agent particle world environments, while \Cref{fig:full-rt-tsc} shows the runtimes for the traffic signal control environments.

\begin{figure*}[ht]
    \centering
    \includegraphics[width=\textwidth]{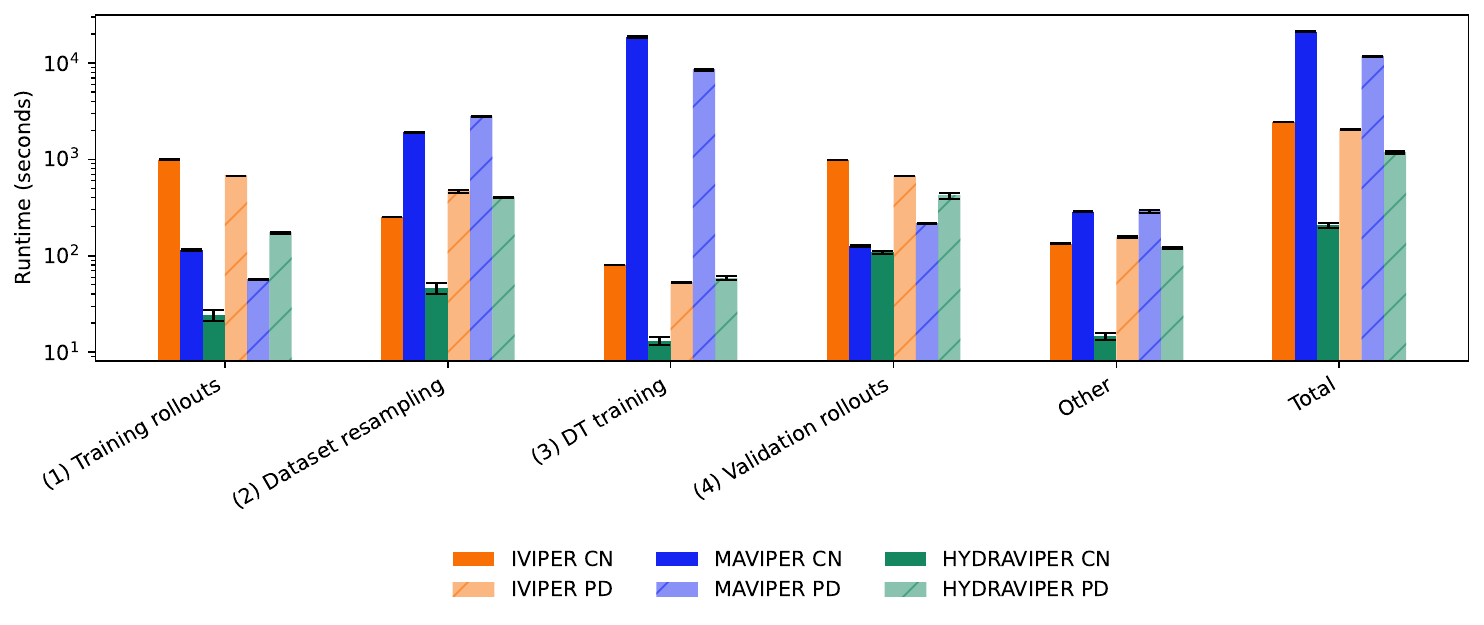}
    \caption{Runtime decomposition for IVIPER, MAVIPER, and HYDRAVIPER (with full environment interaction budget) on the multi-agent particle world environments, cooperative navigation and physical deception. Error bars show 95\% confidence intervals based on 10 random seeds. Runtime measurement was performed following the methodology in \Cref{sec:exp:baselines}.}
    \label{fig:full-rt-mpe}
\end{figure*}

\begin{figure*}[ht]
    \centering
    \includegraphics[width=\textwidth]{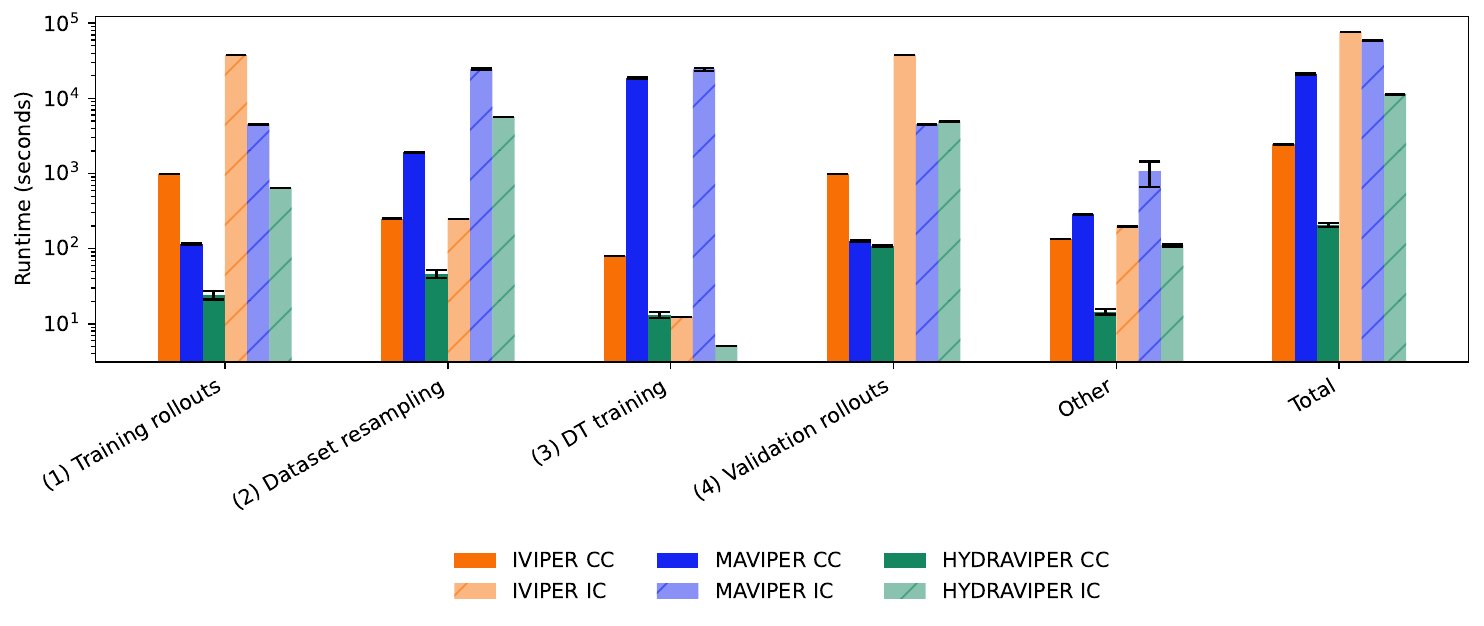}
    \caption{Runtime decomposition for IVIPER, MAVIPER, and HYDRAVIPER (with full environment interaction budget) on the traffic signal control environments, Cologne corridor and Ingolstadt corridor. Error bars show 95\% confidence intervals based on 10 random seeds. Runtime measurement was performed following the methodology in \Cref{sec:exp:baselines}.}
    \label{fig:full-rt-tsc}
\end{figure*}

\section{Full Hyperparameter Sensitivity Results}
\label{sec:app:full-hyp}
In \Cref{fig:hyp-cn}, \Cref{fig:hyp-pd}, \Cref{fig:hyp-cc}, and \Cref{fig:hyp-ic}, we show the performance of HYDRAVIPER in the low budget setting as two hyperparameters --- {\dt} depth and the UCB scaling constant $c$ --- are varied for the cooperative navigation (CN), physical deception (PD), Cologne corridor (CC), and Ingolstadt corridor (IC) environments. We use 500 training/1\,500 validation rollouts for CN and PD, and 100 training/100 validation rollouts for CC and IC.

\begin{figure*}[t]
    \centering
    \begin{subfigure}[ht]{.475\linewidth}
    \includegraphics[width=\linewidth]{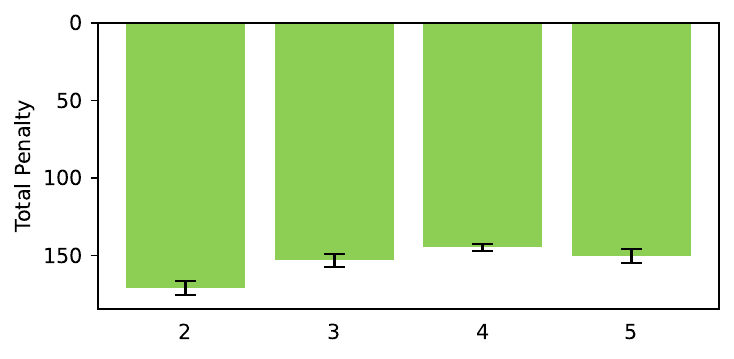}
    \end{subfigure}
    \begin{subfigure}[ht]{.475\linewidth}
    \includegraphics[width=\linewidth]{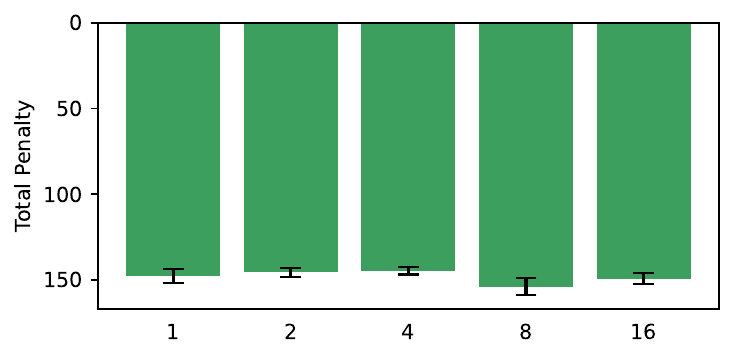}
    \end{subfigure}
    \caption{Sensitivity of HYDRAVIPER to two hyperparameters, {\dt} depth and the UCB scaling constant $c$, on the cooperative navigation (CN) environment. Lower rewards are better.}   
    \label{fig:hyp-cn}
\end{figure*}

\begin{figure*}[t]
    \centering
    \begin{subfigure}[ht]{.475\linewidth}
    \includegraphics[width=\linewidth]{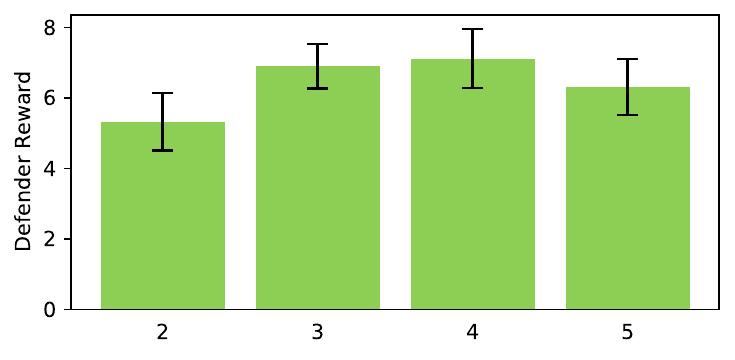}
    \end{subfigure}
    \begin{subfigure}[ht]{.475\linewidth}
    \includegraphics[width=\linewidth]{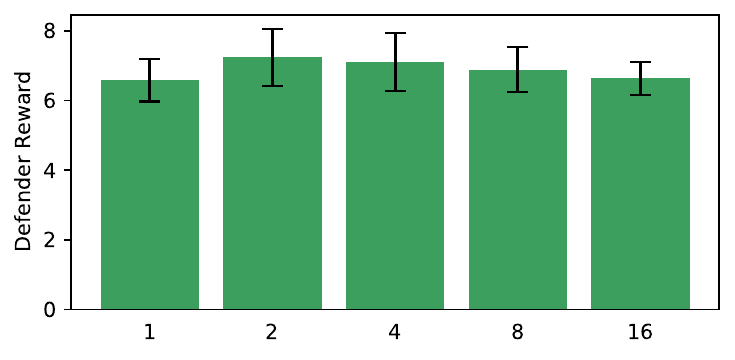}
    \end{subfigure}
    \caption{Sensitivity of HYDRAVIPER to two hyperparameters, {\dt} depth and the UCB scaling constant $c$, on the physical deception (PD) environment. Higher rewards are better.}    
    \label{fig:hyp-pd}
\end{figure*}

\begin{figure*}[t]
    \centering
    \begin{subfigure}[ht]{.475\linewidth}
    \includegraphics[width=\linewidth]{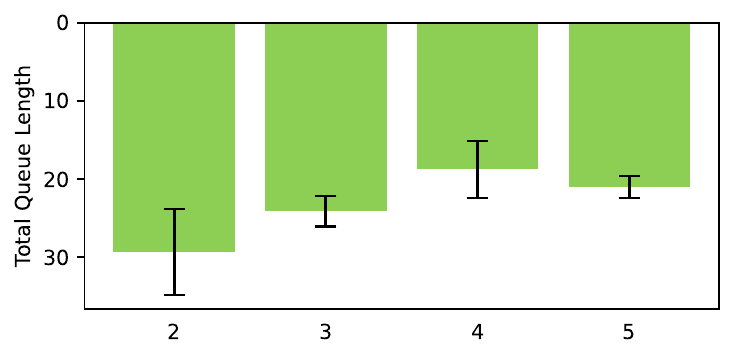}
    \end{subfigure}
    \begin{subfigure}[ht]{.475\linewidth}
    \includegraphics[width=\linewidth]{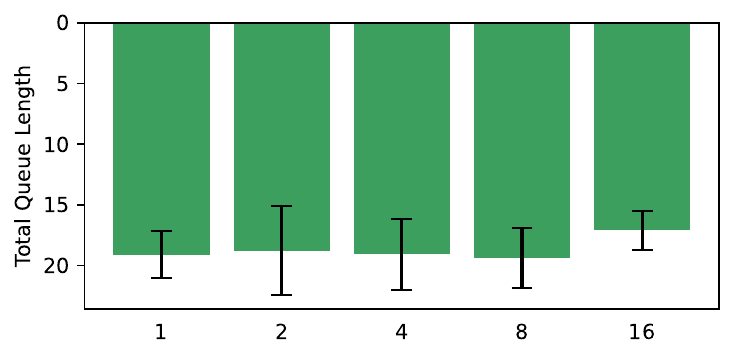}
    \end{subfigure}
    \caption{Sensitivity of HYDRAVIPER to two hyperparameters, {\dt} depth and the UCB scaling constant $c$, on the Cologne corridor (CC) environment. Lower rewards are better.}   
    \label{fig:hyp-cc}
\end{figure*}

\begin{figure*}[t]
    \centering
    \begin{subfigure}[ht]{.475\linewidth}
    \includegraphics[width=\linewidth]{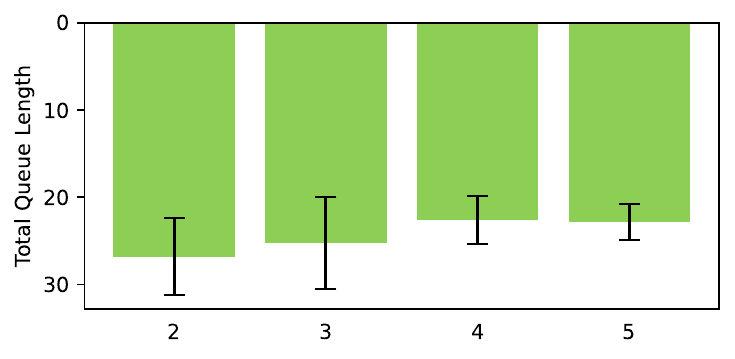}
    \end{subfigure}
    \begin{subfigure}[ht]{.475\linewidth}
    \includegraphics[width=\linewidth]{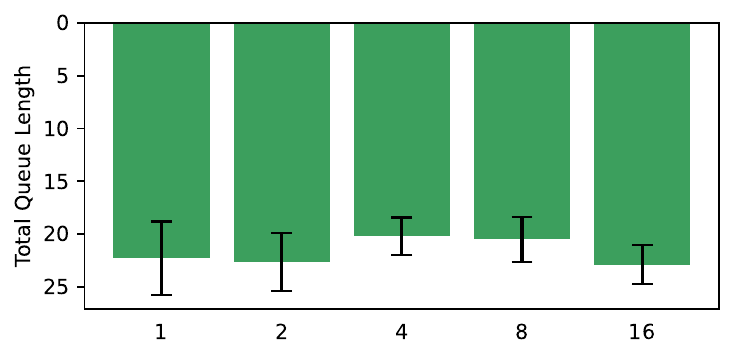}
    \end{subfigure}
    \caption{Sensitivity of HYDRAVIPER to two hyperparameters, {\dt} depth and the UCB scaling constant $c$, on the Ingolstadt corridor (IC) environment. Lower rewards are better.}   
    \label{fig:hyp-ic}
\end{figure*}

\section{Full Ablation Results}
\label{sec:app:full-ablation}
In \Cref{fig:ablation-pd-ic}, we show the performance of HYDRAVIPER in the physical deception (PD, left) and Ingolstadt corridor (IC, right) environments, in an ablation study where each of the algorithm's components --- centralised-$Q$ resampling (C-Q), adaptive training budget allocation (TR-A), and UCB-based validation budget allocation (VR-UCB) --- is removed in turn. For the high budget level, we use 5\,000 training/5\,000 validation rollouts for PD, and 1\,000 training/1\,000 validation rollouts for IC; for the low budget level, we use 500 training/1\,500 validation rollouts for PD, and 100 training/100 validation rollouts for IC.

\begin{figure*}[t]
    \centering
    \begin{subfigure}[ht]{.475\linewidth}
    \includegraphics[width=\linewidth]{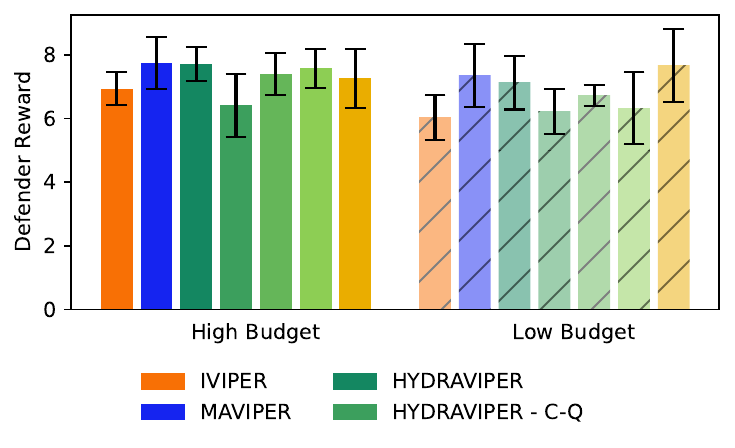}
    \end{subfigure}
    \begin{subfigure}[ht]{.475\linewidth}
    \includegraphics[width=\linewidth]{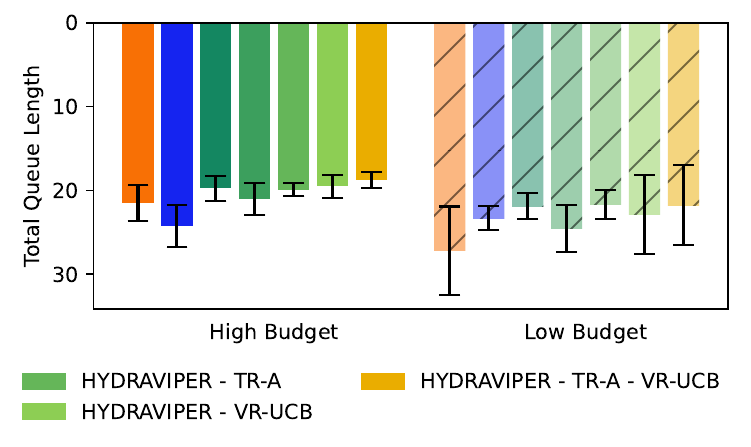}
    \end{subfigure}
    \caption{Ablation on physical deception (PD, left, higher rewards are better) and Ingolstadt corridor (IC, right, lower rewards are better).}
    \label{fig:ablation-pd-ic}
\end{figure*}
\end{document}

%% file: notation.tex
\newcommand{\nagents}{N}

\newcommand{\transitionprob}{P}

\newcommand{\aindex}{i}
\newcommand{\rew}{r}

\newcommand{\ob}{o}
\newcommand{\act}{a}
\newcommand{\s}{s} 
\newcommand{\jointob}{\mathbf{\ob}}
\newcommand{\jointact}{\mathbf{\act}}

\newcommand{\obsspace}{\mathcal{O}}
\newcommand{\actspace}{\mathcal{A}}
\newcommand{\statespace}{\mathcal{S}}
\newcommand{\obsfunc}{O}
\newcommand{\rewardfunc}{R}
\newcommand{\return}{\mathcal{R}}
\newcommand{\meanreturn}{\bar{\rewardfunc}}

\newcommand{\team}{\mathcal{T}}

\newcommand{\agentreward}{\rew_\aindex}
\newcommand{\agentobsspace}{\obsspace_\aindex}
\newcommand{\allobs}{\obsspace_1, ... , \obsspace_\nagents}
\newcommand{\allacts}{\actspace_1, ..., \actspace_\nagents}

\newcommand{\agentactspace}{\actspace_\aindex}

\newcommand{\dt}{DT}
\newcommand{\nn}{NN}

\newcommand{\Xcomment}[1]{{}}

\newcommand{\noaccents}[1]{#1}

\newcommand{\newagentvar}[3][\noaccents]{%
\expandafter\newcommand\expandafter{\csname #2\endcsname}{#1{#3}}%
\expandafter\newcommand\expandafter{\csname #2s\endcsname}{#1{\boldsymbol{#3}}}%
\expandafter\newcommand\expandafter{\csname #2smi\endcsname}[1][i]{#1{\boldsymbol{#3}}_{-##1}}%
\expandafter\newcommand\expandafter{\csname #2i\endcsname}[1][i]{#1{#3}_{##1}}%
\expandafter\newcommand\expandafter{\csname #2ith\endcsname}[1][i]{#1{#3}_{(##1)}}%
}

\DeclareMathOperator*{\argmax}{argmax}

\newagentvar{val}{v}
\newagentvar[\tilde]{mulval}{v}
\newagentvar[\tilde]{multval}{v}
\newagentvar[\hat]{divval}{v}

\newcommand{\ceil}[1]{\left\lceil{#1}\right\rceil}



\algnewcommand{\IIf}[1]{\State\algorithmicif\ #1\ \algorithmicthen}
\algnewcommand{\EndIIf}{\unskip\ \algorithmicend\ \algorithmicif}
\makeatletter
\algnewcommand{\LineComment}[1]{\Statex \hskip\ALG@thistlm \(\triangleright\) #1}
\algnewcommand{\IndentLineComment}[1]{\Statex \hskip\ALG@tlm \(\triangleright\) #1}
\makeatother

{
\setlength{\topsep}{0.5em}
\newtheoremstyle{theorembf}
  {\topsep}
  {\topsep}
  {\itshape}
  {0pt}
  {\bfseries\scshape}
  {.}
  { }
  {\thmname{#1}\thmnumber{ #2}\textnormal{\thmnote{ (#3)}}}
\newtheorem{rsq}{RQ}%

\newtheoremstyle{theoremit}
  {\topsep}
  {\topsep}
  {\itshape}
  {0pt}
  {\scshape}
  {.}
  { }
  {\thmname{#1}\thmnumber{ #2}\textnormal{\thmnote{ (#3)}}}

\newtheorem{corollary}{Corollary}%

}
\crefname{rsq}{RQ}{RQs}
\Crefname{rsq}{RQ}{RQs}
\crefname{hyp}{hypothesis}{hypotheses}
\Crefname{hyp}{Hypothesis}{Hypotheses}
\crefname{corollary}{corollary}{corollaries}
\Crefname{corollary}{Corollary}{Corollaries}